\documentclass[10pt,journal,a4paper]{IEEEtran}
\usepackage{amssymb,amsmath}
\usepackage{cite}
\usepackage{graphicx,subfigure}
\usepackage{psfrag}
\usepackage{url}
\usepackage[latin1]{inputenc}
\usepackage[absolute,overlay]{textpos}

\usepackage{pgf}
\usepackage{stackengine}
\newcommand\stackequal[2]{%
	\mathrel{\stackunder[2pt]{\stackon[4pt]{$\gtrless$}{$\scriptscriptstyle#1$}}{%
			$\scriptscriptstyle#2$}}}
\usepackage{booktabs}
\usepackage{algorithm}
\usepackage{algpseudocode}
\usepackage[export]{adjustbox}

\usepackage{amsthm}

\newtheorem{proposition}{Proposition}
\newtheorem{lemma}{Lemma}
\newtheorem{theorem}{Theorem}
\newtheorem{corollary}{Corollary}
\newtheorem{remark}{Remark}

\begin{document}

\title{Distributed Rate Control in Downlink NOMA Networks with Reliability Constraints}

	\author{
		\IEEEauthorblockN{Onel L. A. López, \IEEEmembership{Student Member, IEEE}, %
			Hirley Alves, \IEEEmembership{Member, IEEE},		
			Matti Latva-aho, \IEEEmembership{Senior Member, IEEE}
		}		%
		
		\thanks{Authors are with the Centre for Wireless Communications (CWC), University of Oulu, Finland. \{onel.alcarazlopez,hirley.alves,matti.latva-aho\}@oulu.fi}
		
		\thanks{This work is partially supported by Academy of Finland 6Genesis Flagship (Grant n.318927, and  n.303532, n.307492) and by the Finnish Foundation for Technology Promotion.}
	}		
					
\maketitle

\begin{abstract}
Non-orthogonal multiple access (NOMA) has been identified as a promising technology for future wireless systems due to its performance gains in spectral efficiency when compared to conventional orthogonal schemes (OMA).
This gain can be easily translated to an increasing number of served users, but imposes a challenge in the system reliability which is of vital importance for new services and applications of coming cellular systems. To cope with these issues we propose a NOMA rate control strategy that makes use only of topological characteristics of the scenario and the reliability constraint. We attain the necessary conditions so that NOMA overcomes the OMA alternative, while we discuss the optimum allocation strategies for the 2-user NOMA setup when operating with equal rate or maximum sum-rate goals. In such scenario we  show that the user with the largest target error probability times the ratio between the average receive signal power and the average interference power, should be scheduled to be decoded first for optimum performance. We compare numerically the performance of our allocation scheme with its ideal counterpart requiring full CSI at the BSs and infinitely long blocklength, and show how the gap increases as the reliability constraint becomes more stringent. Results also evidence the benefits of NOMA when the co-interference can be efficiently canceled, specially when the goal is to maximize the sum-rate.
\end{abstract}
\begin{IEEEkeywords}
	NOMA, rate control, power allocation, reliability constraints, optimal user ordering.
\end{IEEEkeywords}
\section{Introduction}
Non-orthogonal multiple access (NOMA) has been widely recognized as a promising technology for future wireless systems due to its superior spectral efficiency compared to conventional orthogonal multiple access (OMA) \cite{Saito.2013,Razavi.2017}.
Compared to OMA, NOMA can exploit the channel diversity more efficiently via smart interference management techniques such as successive interference cancellation (SIC) \cite{Ding.2017}, thus, achieving the performance enhancements. A unified framework  for NOMA is provided in \cite{Wang.2018} where the authors review the principles of various NOMA schemes in different domains. As one of the dominant NOMA schemes, power-domain NOMA, where different users are allocated different power levels according to their channel conditions, has become a strong candidate.
These setups are based on the combination of superposition coding at the transmitter side, by allocating the same frequency/time/spatial resource to multiple receivers and multiplexing them on the power domain, and then extracting the intended signals from the composite data using SIC at these receivers \cite{Saito.2013}. Power domain NOMA has received extensive research interests and readers can refer to \cite{Islam.2017} for a comprehensive survey on its recent progress towards 5G systems and beyond. 
Particularly, the work in \cite{Zhang.2018} studies the tradeoff between data rate performance and energy consumption by examining the problem of energy-efficient user scheduling and power optimization in 5G NOMA heterogeneous networks. The scenario of NOMA-based cooperative relay network is explored in \cite{Wan.2018}, while coordinated multipoint transmissions in  downlink NOMA cellular systems are analyzed in \cite{Ali.2018}. 
\subsection{Related Works}\label{IntA}
 Many of the existing studies focus on the system performance when considering a single NOMA cell/cluster \cite{Ding.2014,Ding.2016,Zhang2.2016,Chen.2017,Kang.2017,Xu.2017,Yu.2018}. The performance of NOMA in a cellular downlink scenario with randomly deployed users is investigated in \cite{Ding.2014}. Authors show that NOMA can achieve superior performance in terms of ergodic sum rates; however, that depends critically on the choices of the users' targeted data rates and allocated power. In \cite{Ding.2016}, the impact of user pairing on the performance of two NOMA systems, NOMA with fixed power allocation and cognitive-radio-inspired NOMA, is characterized. An uplink power control scheme is proposed in \cite{Zhang2.2016} in order to achieve diverse arrived power in an uplink NOMA system, while authors analyze its outage performance and the achievable sum data rate. Interestingly, authors in \cite{Chen.2017} investigate the theoretical performance comparison between NOMA and conventional OMA systems from an optimization point of view, while in the context of machine type communications (MTC) authors in \cite{Kang.2017} propose a game theory power control algorithm that prioritizes first the  communication reliability and once the reliability is satisfied it focuses on the power consumption issues. Also,  energy-efficient NOMA design problem for two downlink receivers that have strict reliability and finite blocklength (latency) constraints have been considered in \cite{Xu.2017}, while in \cite{Yu.2018} the physical-layer transmission latency reduction enabled by NOMA in short-packet communications has been highlighted. Finally, the combination of NOMA  with the multiple-input multiple-output (MIMO) technology is explored in \cite{Huang.2018} along with the limitations and future research directions in the area.

Above works do not deal with inter-cell interference, which is a pervasive problem in most of the existing wireless networks and could significantly limit the performance of NOMA deployments. In that sense some recent research works have considered the performance characterization of large-scale NOMA systems \cite{Tabassum.2017,Lopez.2018,Lopez2.2018,Di.2017,Zhang.2016,Liu.2017,Salehi.2019}. Using stochastic geometry, the performance of uplink NOMA in terms of the rate coverage and average achievable rate is characterized in \cite{Tabassum.2017} using Poisson cluster process and considering both perfect and imperfect SIC. Additionally, authors in \cite{Lopez.2018} proposes and evaluate NOMA as a massive-MTC enabler technology in data aggregation networks, while identifying the power constraints on the devices sharing the same channel in order to attain a fair coexistence with purely OMA setups. As a continuation, the relaying phase, where the aggregatted data is forwarded to the base station, is considered in \cite{Lopez2.2018}, and the system performance is investigated in terms of average number of users that are simultaneously served. In the context of vehicle-to-everything (V2X) communications, NOMA has been explored in \cite{Di.2017} where the Base Station (BS) performs semi-persistent scheduling and allocates time-frequency resources in a non-orthogonal manner, while the vehicles autonomously perform distributed power control when broadcasting their safety information to the neighborhood. System coverage and average achievable rate of the $m$-th rank user is derived in \cite{Zhang.2016} assuming that the BSs locations follow a homogeneous Poisson Point Process (PPP) in downlink NOMA, and authors show that NOMA can bring considerable performance gain compared to OMA when SIC error is low. Also, the performance of two-user downlink NOMA is investigated \cite{Liu.2017} in a multi-tier cellular network where the macro cell BSs use the massive MIMO technology and each small cell adopts user pairing to implement two-user NOMA transmission. 
Finally, authors in \cite{Salehi.2019} developed an analytical framework to derive the meta distribution of the Signal-to-Interference Ratio (SIR) in large-scale co-channel uplink and downlink NOMA networks with one NOMA cluster per cell. Notice that such meta distribution tool provides a more precise characterization of the performance of a typical transmission link than the standard success probability, since it is able of characterizing the fraction of users that perform with a given reliability, whereas the standard success probability just characterizes the fraction of users that are in coverage.

On the other hand, NOMA works have usually been aimed at improving the spectral \cite{Ding.2014,Ding.2016,Zhang2.2016,Chen.2017} and energy \cite{Xu.2017,Liu.2017} efficiency of the system, the coverage probability \cite{Tabassum.2017,Zhang.2016,Liu.2017}, and also recently have focused in providing massive connectivity \cite{Lopez.2018,Lopez2.2018} for future Internet of Things (IoT) scenarios. Besides the massive access problem, many IoT use cases with stringent delay and reliability constraints, e.g., ultra reliable MTC (uMTC), have been identified for the coming years and constitute important challenges for emerging massive wireless networks. In that regard, the works \cite{Kang.2017,Xu.2017,Yu.2018} analyze reliability and delay metrics but in a single cell setup, while  authors in \cite{Di.2017} do evaluate a multi-cell setup but in a V2X network. By using stochastic geometry tools, authors in \cite{Salehi.2019} maximize the success probability with and without latency constraints in uplink and downlink NOMA setups with perfect SIC. 
Additionally, optimizing the performance of these networks has usually relied on power control, sub-channel allocation and user selection and/or user pairing mechanisms, while the transmission rate is another degree of freedom that could be also exploited. In fact, the latter has been already explored efficiently in ad-hoc \cite{Kalamkar} and OMA \cite{Lopez3.2018,Lopez.2019} networks with reliability constraints, where the proposed rate allocation schemes are based on easy-to-obtain topological characteristics and on the (Rayleigh) fading statistics\footnote{Specifically, \cite{Kalamkar} focuses on an interference-limited Poisson bipolar network, while in \cite{Lopez3.2018,Lopez.2019} we did similar but for cellular networks with multiple antennas at the receiver side.}. Notice that rate allocation mechanisms are promising since their impact is merely local, in the sense that a node varying its transmission rate does not influence in the performance of any other node\footnote{In case of NOMA, it would affect the nodes associated to the same NOMA cluster, but this is still local and easy to handle.}, thus, when properly designed they are suitable for distributed implementations.
\subsection{Contributions and Organization}\label{IntB}
To the best of our knowledge, a multi-cell downlink NOMA scenario where User Equipments (UEs) operate with reliability constraints and imperfect SIC has not been explored yet, thus, in this work we aim at filling that gap. Our main contributions are
\begin{itemize}
	\item We propose a NOMA rate allocation scheme to be performed by BSs in order to meet the reliability constraints of their associated UEs. The advantage of this scheme is that no instantaneous Channel State Information (CSI) is required at the transmission side, while it only relies on easy-to-get information such as the average receive power from the desired signal and from the interfering BSs at the target UE, and the reliability constraint. Notice that this scheme is closely related to the one we first proposed in \cite{Lopez3.2018} and extended in \cite{Lopez.2019}, but there we just analyze an OMA network with multiple antennas at the receiver, while operating in a non-orthogonal fashion raises many other practical concerns;
	\item We characterize analytically the distribution of the allocated SIR threshold, thus, the distribution of the allocated rate, in Poisson Networks;
	\item We attain the necessary conditions so that NOMA overcomes the OMA alternative, while we discuss the optimum allocation strategies for the 2-UEs NOMA setup when operating with equal rate or maximum sum-rate goals. Specifically, we find the optimum rate and power allocation profile and  show that the UE with the largest value of target error probability times the ratio between the average receive power from the desired signal and the average interference power, should be scheduled to be decoded first for optimum performance;
	\item We compare numerically the performance of our allocation scheme with its ideal counterpart requiring full CSI at the BSs and infinitely long blocklength, and show how the gap increases as the reliability constraint becomes more stringent. Results also evince the benefits of NOMA when the co-interference can be efficiently canceled, specially when the goal is to maximize the sum-rate.
\end{itemize}

Next, Section~\ref{system} establishes the system model and assumptions. In Section \ref{FB} we present the rate allocation scheme, and characterize the distribution of the adopted SIR threshold in large-scale networks. Section~\ref{power} discusses the optimum power allocation strategy in the 2-UEs NOMA setup, while Section \ref{results} presents the numerical results. Finally, Section \ref{conclusions} concludes the paper.

\textit{Notation:} 
 $\mathbb{E}[\!\ \cdot\ \!]$ and $\mathbb{E}[\!\ \cdot\ \!|A]$ denote expectation and expectation conditioned on event $A$, respectively, while $\Pr(B)$ and  $\Pr(B|A)$ are the probability of event $B$, and $\Pr(B)$ conditioned on $A$, respectively.
 $f_X(x)$ and $F_X(x)$ are the Probability Density Function (PDF) and Cumulative Distribution Function (CDF) of random variable (RV) $X$, respectively. $X\sim\mathrm{Exp}(1)$ is an exponential distributed RV with unit mean, e.g., $f_X(x)=\exp(-x)$ and $F_X(x)=1-\exp(-x)$; while $f_{r_i}(x)=2\pi\lambda x\exp({-\pi\lambda x^2})$ denotes the PDF of the nearest neighbor distance in a two-dimensional PPP with density $\lambda$, which is given in \cite{Haenggi.2005}. $d^{m} f(x)/d x^m$ is the $m$-th derivative of function $f(x)$ with respect to $x$, while $\min(x,y)$ is the minimum between the values of $x$ and $y$. Finally, $\mathcal{L}_X(s)$ denotes the Laplace transform of RV $X$, while $\mathcal{L}^{-1}\{\cdot\}(x)$ is the inverse Laplace transform operation.
\section{System Model}\label{system}
Consider a multi-cell downlink cellular network where BSs are spatially distributed
according to a 2-D homogeneous PPP $\Phi$ with density $\lambda$. 
We assume each UE is associated with the closest BS, namely the UEs in the Voronoi cell of a BS are associated with it, resulting in coverage areas as shown in Fig.~\ref{Fig1}.
UEs are placed such that the distance between $\mathrm{UE}_i$ and its associated BS, is $r_i$.
We also consider an interference-limited wireless system given a dense deployment of small cells and hence the impact of noise is neglected throughout the paper\footnote{However, the impact of the noise could easily be incorporated without substantial changes.}. All BSs transmit at the power $P_T$.  We adopt a channel model that comprises standard path-loss with exponent $\alpha$ and Rayleigh fading.
\begin{figure}[t!]
	\includegraphics[width=0.48\textwidth,center]{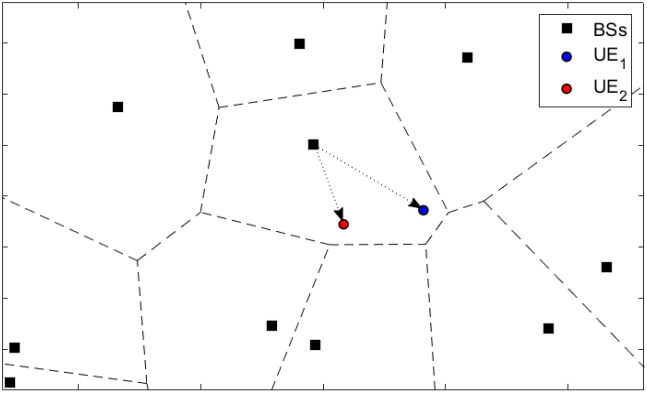}
	\caption{Illustration of the system model for $M=2$.}		
	\label{Fig1}
\end{figure}
 
A NOMA group, $\mathcal{M}$, consists of $M$ UEs, $\mathrm{UE}_i,\ i=1,2,...,M$, and notice that an illustration of the system model for $M=2$ is shown in Fig.~\ref{Fig1}. Each $\mathrm{UE}_i$ tries to decode the interfering NOMA signals destined to $\mathrm{UE}_m$, $m\le i$, and then remove the interfering NOMA signals from the received signal, in a successive manner. Assume the transmission power of signal for $\mathrm{UE}_i$ is $P_i$ with $\sum_{i\in \mathcal{M}}P_i=P_T$. Our system model is similar to \cite{Zhang.2016}, although, we do not establish the decoding order based on the channel gains normalized by the interference since such approach requires full knowledge not only of the channel coefficients and path-loss but also of the interference experienced each time at each $\mathrm{UE}$, which is very difficult to obtain. Instead, the BS allocates the transmit power such that $P_1\ge...\ge P_M$ for guaranteeing the desired decoding order\footnote{For instance, if two $\mathrm{UE}s$, named $\mathrm{UE}_A$ and $\mathrm{UE}_B$, are served such that the BS allocates $P_B>P_A$, then $\mathrm{UE}_B$ and $\mathrm{UE}_A$ signals are decoded in that order at the $\mathrm{UE}s$ since for each of them the signal associated with $\mathrm{UE}_B$ arrives always with greater power. Thus, $\mathrm{UE}_B=\mathrm{UE}_1$ and $\mathrm{UE}_A=\mathrm{UE}_2$.} and some criteria for feasibility and optimum performance are provided in Section~\ref{power}. Finally, CSI is only required at the UEs and not at transmitter side.
\subsection{Signal Model} 
We denote the signal intended to $\mathrm{UE}_i$ as $x_i$ where $\mathbb{E}[|x_i|^2]=1$. According to NOMA principle \cite{Razavi.2017}, the transmitted signal at BS is coded as the composite signal from all the $\mathrm{UE}$s sharing the orthogonal resource,
\begin{align}
x=\sum\limits_{i=1}^M\sqrt{P_i}x_i.\label{signal}
\end{align}
Thus, in an interference limited system the received signal at $\mathrm{UE}_i$ can be represented as
\begin{align}
y_i=\sqrt{h_ir_i^{-\alpha}}x+I_i,\label{y}
\end{align}
where $h_i\sim\mathrm{Exp}(1)$ is the power channel gain coefficient, and
\begin{align}
I_i=\sum_{j\in\Phi\backslash \{b_0\}}g_jr_{j,i}^{-\alpha}P_T \label{Ii}
\end{align}
is the inter-cell interference, which is the sum of the power received from all BSs excluding the serving BS denoted as $b_0$. In \eqref{Ii}, $g_j\sim \mathrm{Exp}(1)$ and $r_{j,i}$ are the Rayleigh fading power coefficient of interfering channel, and transmission distance from interfering BS $j$ to $\mathrm{UE}_i$, respectively.
\subsection{SIR after SIC}
Instead of the real SIR at the receiving antenna of a $\mathrm{UE}$, we are more interested in the SIR after SIC which is directly related to the performance of each particular $\mathrm{UE}$. Notice that $\mathrm{UE}_1$ does not need to perform interference cancellation and directly treats $x_j, j\ge 2$, as interference since it comes the first in the decoding order. On the other hand, $\mathrm{UE}_2$ first decodes $x_1$ and removes it from the received composite signal $y_2$, based on which $\mathrm{UE}_2$ can further decode $x_2$. Thus, $\mathrm{UE}_m$ first decodes signals $x_j,j<m$, and remove them from the received composite signal $y_m$, based on which $\mathrm{UE}_m$ can further decode $x_m$. Assuming successful decoding but with error propagation \cite{Sun.2016} given by the parameter $\mu\in[0,\ 1]$ measuring the imperfection of SIC, the NOMA SIR of the $i-$th signal at $\mathrm{UE}_i$ can be expressed as
\begin{align}\label{eq4}
\mathrm{SIR}_i=\frac{h_ir_i^{-\alpha}P_i}{h_ir_i^{-\alpha}\Big[\mu\sum\limits_{j=1}^{i-1}P_j+\sum\limits_{j=i+1}^{M}P_j\Big]+I_i}.
\end{align}
Notice that $\sum\limits_{j=1}^{i-1}P_j=0$ and $\sum\limits_{j=i+1}^{M}P_j=0$ for $i=1$ and $i=M$, respectively.
\section{Distributed rate control under reliability constraint}\label{FB}
We are interested in finding the SIR threshold for each link $i\in\mathcal{M}$, $\gamma_i$, such that the conditional link success probability, $\mathbb{P}(\mathrm{SIR}_i>\gamma_i|\Phi)$ \cite{Kalamkar}, is equal to $1-\epsilon_i$, then, the transmitter decides on its rate as $R_i=\log_2(1 + \gamma_i)$. Allocating the rate in such way guarantees that all the $\mathrm{UE}$s meet their reliability requirements since $\mathbb{P}(\mathrm{SIR}_i>\gamma_i|\Phi)$ is defined for each network realization. Notice that for practical feasibility the allocation should rely on information easy to obtain at the BS; for instance, allocating the rate based on the path-loss experienced by each of the interfering signals is not suitable.
\begin{lemma}\label{the1}
	For a given realization of $\Phi$, the maximum SIR threshold $\gamma_i$ that guarantees the required reliability for the $i-$th link is 
	\begin{align}
	\gamma_i=\frac{\varphi^*_iP_i}{P_T+\varphi^*_i\Big(\mu\sum\limits_{j=1}^{i-1}P_j+\sum\limits_{j=i+1}^{M}P_j\Big)},\label{gami}
	\end{align}
	where $\varphi^*_i$ is the unique real positive solution of	
	\begin{align}
    \prod_{j\in \Phi\backslash \{b_0\}}(1+\varphi_i r_i^{\alpha}r_{j,i}^{-\alpha})=\frac{1}{1-\epsilon_i}.\label{eqP}
    \end{align}	
\end{lemma}
\begin{proof}
See Appendix~\ref{App_A}. \phantom\qedhere
\end{proof}
\begin{remark}\label{re1}
	If $M=1$ we have that $\gamma_1=\varphi_1^*$, thus, $\varphi_i^*$ represents the SIR threshold that would be required for $\mathrm{UE}_i$ if it was operating alone in the channel, hence $\varphi_i^*\ge\gamma_i$ holds always. Additionally,
	\begin{align}
	\frac{d \gamma_i}{d \varphi_i^*}=\frac{P_iP_T}{\bigg(P_T+\varphi^*_i\Big(\mu\sum\limits_{j=1}^{i-1}P_j+\sum\limits_{j=i+1}^{M}P_j\Big)\bigg)^2}>0,\label{d2}
	\end{align}
	therefore, $\gamma_i$ is an increasing function of $\varphi_i^*$.
\end{remark}

For practical values of $\epsilon_i$ we next propose a very tight and simple approximation for $\varphi_i^*$ that is suitable when computing \eqref{gami} for rate allocation.
\begin{theorem}\label{the2}
	$\varphi_i^*$ approximates accurately to
	\begin{align}
	\varphi_i^*&\approx \frac{r_i^{-\alpha}}{\sum_{j\in\Phi\backslash\{b_0\}}r_{j,i}^{-\alpha}}\epsilon_i\label{app}
	\end{align}
	when $\epsilon_i\le 10^{-1}$.
\end{theorem}
\begin{proof}
See Appendix~\ref{App_B}. \phantom\qedhere
\end{proof}
Notice that our results are not restrictive to the adopted path-loss model. In fact, $r_i ^{-\alpha}/\sum\limits_{j\in\Phi\backslash\{b_0\}}r_{j,i}^{-\alpha}$ can be expressed in a generalized way as the quotient between the average receive power from the desired signal and from the interfering BSs at the given location, and in that aspect lies the main advantage of using \eqref{app}. Fig.~\ref{Fig2}a corroborates the accuracy of \eqref{app}, which is independent of the specific deployment and path-loss exponent. As expected, the less strict the reliability constraint, the larger $\varphi_i^*$, hence the larger $\gamma_i$.
\begin{figure}[t!]
	\includegraphics[width=0.47\textwidth,right]{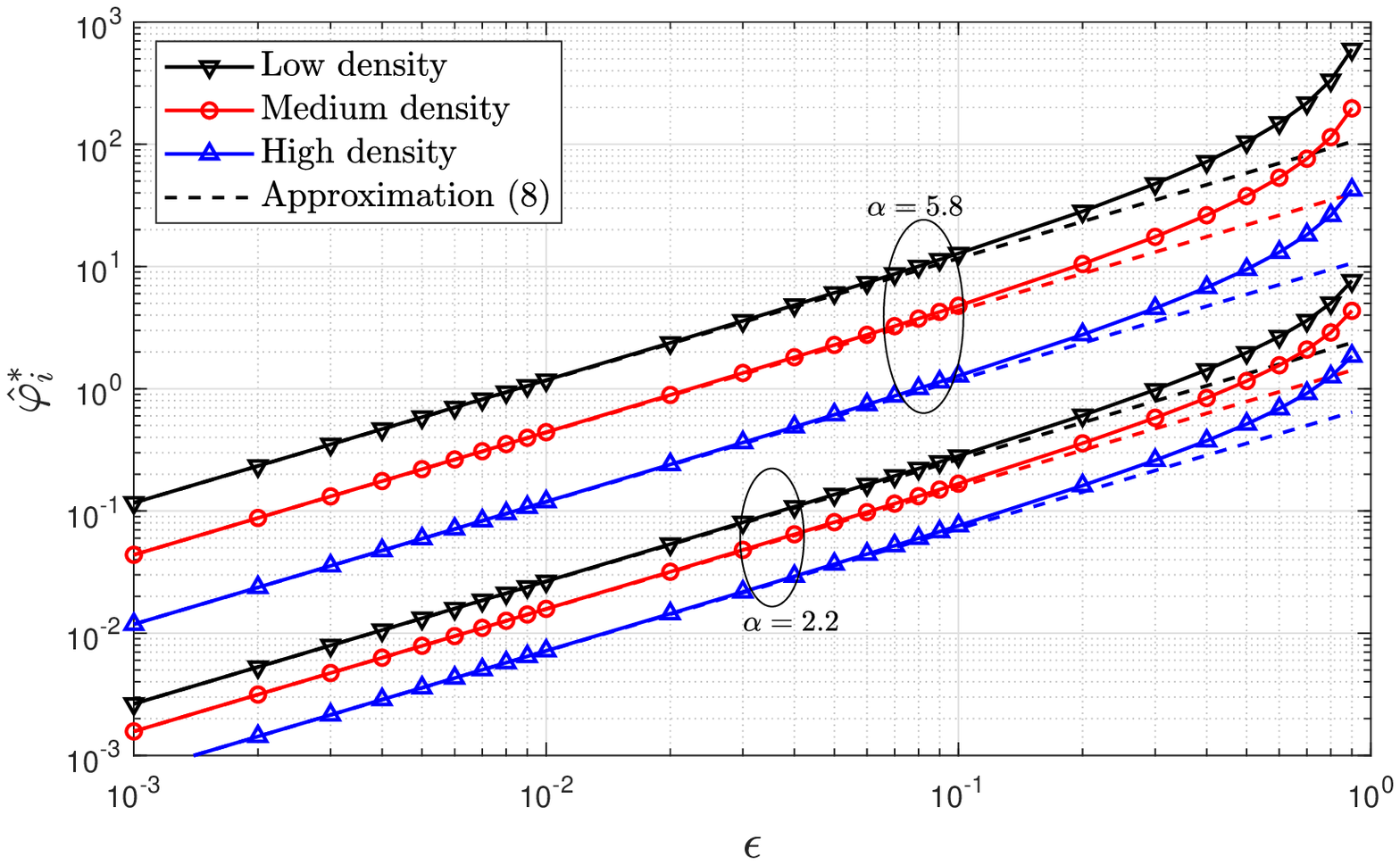}\vspace{2mm}\\ 
    \includegraphics[width=0.47\textwidth,right]{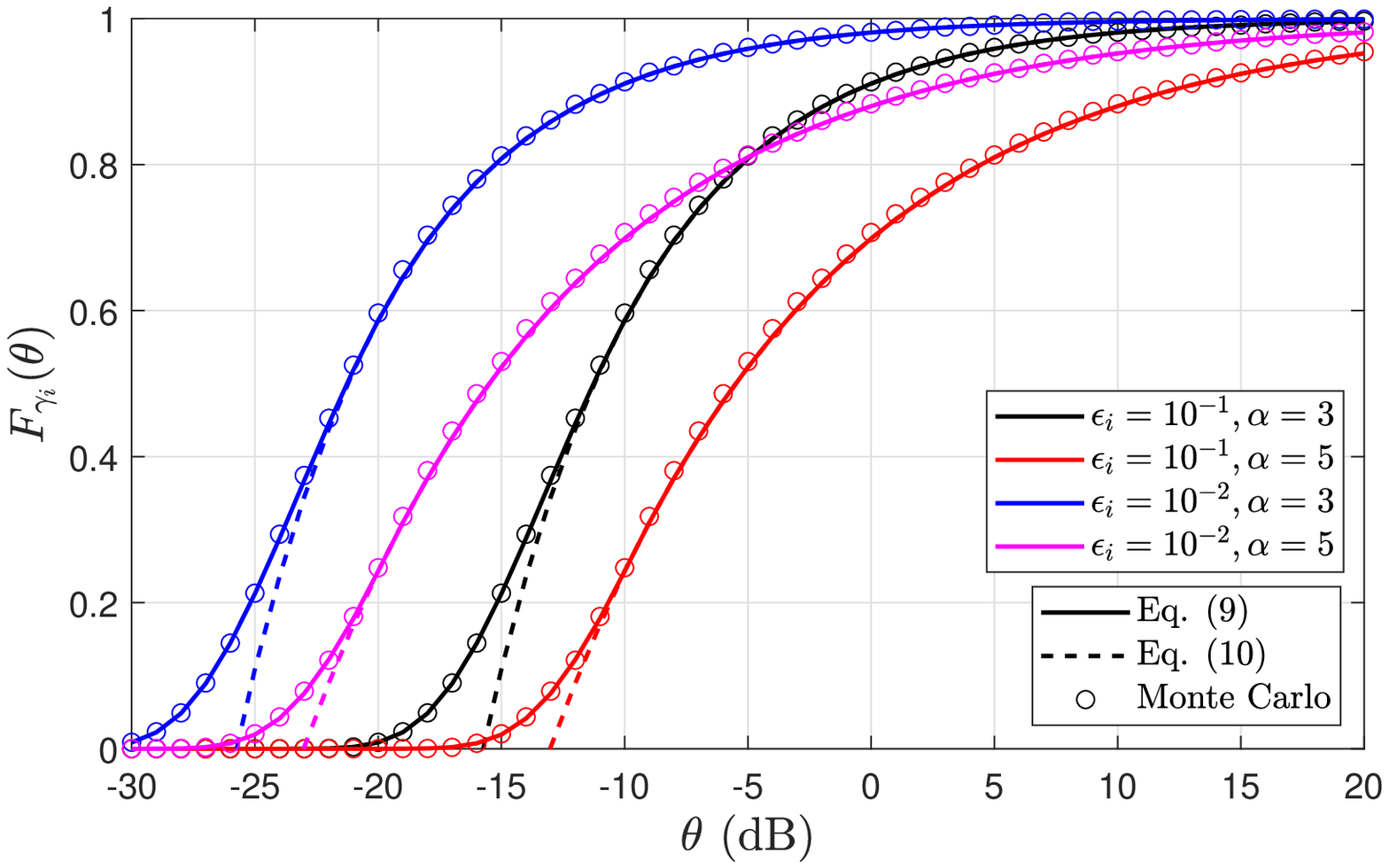}
	\caption{(a) $\varphi_i^*$ as a function of $\epsilon_i$, for high density, $r_{j,i}=40+10j$, medium density, $r_{j,i}=40+20j$ and low density,  $r_{j,i}=40+30j$, example deployments (top). We set $r_i=30$m. (b) $F_{\gamma_i}(\theta)$  for OMA scheme ($P_i=P_T$ and $\gamma_i=\varphi_i^*$) (bottom). Since $F_{\gamma_i}(\theta)$ does not depend on $\lambda$ we choose an arbitrary value of $\lambda=10^{-4}/\mathrm{m}^2$ ($100/\mathrm{km}^2$) for obtaining the Monte Carlo results.}
	\label{Fig2}
\end{figure}
\subsection{Distribution of the SIR threshold}\label{dSIR}
Herein we find the distribution of the SIR threshold allocated to each UE for satisfying its reliability constraint. The distribution is taken with respect to the PPP spatial randomness and the results are useful for characterizing the links' rate performance in general large-scale networks.
\begin{theorem}\label{the3}
	The distribution of the SIR threshold allocated to every $\mathrm{UE}_i$  for satisfying its reliability constraint $\epsilon_i$ according to \eqref{gami} and \eqref{app}, is given by
	\begin{align}
	F_{\gamma_i}(\theta&)=1-\mathcal{L}^{-1}\bigg\{\frac{1}{s\ _1F_1(-\delta,1-\delta,-s)}\bigg\}\big(z_i(\theta)\big), \label{eqth3}\\
	& \mathrel{\mathop{\leqq}^{z_i(\theta)>1}_{z_i(\theta)\le 1}} 1-\mathrm{sinc}(\delta) z_i(\theta)^{\delta}, \label{eqth4}
	\end{align}
    in downlink NOMA networks with PPP distributed BSs. In \eqref{eqth3} and \eqref{eqth4} we use $\delta=2/\alpha$ and
	\begin{align}
	z_i(\theta)&=\left( \frac{P_i}{\theta}-\mu\sum\limits_{j=1}^{i-1}P_j-\sum\limits_{j=i+1}^{M}P_j \right) \frac{\epsilon_i}{P_T}.\label{ztheta}
	\end{align}
\end{theorem}
\begin{proof}
    See Appendix~\ref{App_C}. \phantom\qedhere	
\end{proof}

Notice that for $z_i(\theta)>1$ there is no closed-form analytical expression for $F_{\gamma_i}(\theta)$ since it is required computing the Laplace inversion of a nasty function which includes an hypergeometric term, thus, numerical methods are required (see \cite{Abate.2006} for algorithmic implementations of numeric Laplace inversion). Still, \eqref{eqth3} is of enormous importance since Monte Carlo simulation of these scenarios are particularly time consuming. Another important observation is that $F_{\gamma_i}(\theta)$ does not depend on the network density $\lambda$ which matches previous results as in \cite[Section~III]{Haenggi.2016} for Poisson cellular networks without rate control and reliability restrictions.
	
Fig.~\ref{Fig2}b corroborates \eqref{eqth3} and \eqref{eqth4} since the numerical evaluation of \eqref{eqth3} and the Monte Carlo results match, while \eqref{eqth4} matches when $\theta\ge\epsilon_i$ and works as a lower bound in the complementary region.
For a fixed reliability constraint it is shown that as the path-loss exponent increases, the curves shift to the right and greater data rates are feasible.
This is because we are considering an interference-limited setup and notice that  the interference is affected more than the receive power of the desired signal itself since interfering signals endure larger distances.
While for fixed path-loss exponent the allocated data rates decrease as the reliability constraints become more strict and in such case the gap between curves with different values of $\epsilon_i$ matches their quotient according to \eqref{app}. For instance, the gap between the curves with $\epsilon_i=10^{-1}$ and $\epsilon_i=10^{-2}$ in Fig.~\ref{Fig2}b is $10$dB which matches $10^{-1}/10^{-2}$.
\section{Power Allocation and Feasibility}\label{power}
Herein, we derive and analyze the main conditions that are necessary for the model previously described to be feasible, while we discuss the optimum power allocation strategy for two transmission rate allocation problems: equal-rate (Subsection~\ref{power_e}) and maximum sum-rate (Subsection~\ref{power_s}). Notice that the transmission rate, $\mathcal{R}$, differs conceptually from the achievable rate, which is a more relevant performance metric since accounts also for the unsuccessful attempts; however, in practical scenarios where $\epsilon_i\le 10^{-1}$ both metrics approximate well. This and the fact that dealing with the transmission rate is analytically easier are the main reasons why we selected it as the performance metric.

For analytical tractability we focus on the $M=2$ setup, but notice that even when some existing results show that NOMA with more devices may
provide a better performance gain \cite{Zhang.2016}, this may not be practical. The reason is that considering processing complexity for SIC receivers, especially when SIC error propagation is considered, 2-users NOMA is actually more practical in reality \cite{Zhang.2017,Liu.2016}.
Finally, let's set $P_1=\beta P_T$, $P_2=(1-\beta)P_T$, where $\beta\in[1/2,\ 1]$ since by NOMA definition and SIC operation: $P_1\ge P_2$. Therefore, parameter $\beta$ completely determines the power allocation profile.
\subsection{Equal-rate allocation}\label{power_e}
Herein we consider the scenario in which both UEs are scheduled with the same transmission rate $\mathcal{R}=\log_2(1+\gamma_1)=\log_2(1+\gamma_2)$. Therefore, maximum fairness is attained which is practically advantageous when UEs require transmissions with very similar QoS constraints, e.g., similar $\epsilon_i$ and data rates, although our results hold independently of this. For this case and according to Remark~\ref{re1}: $\gamma_1=\gamma_2<\min(\varphi_1^*,\varphi_2^*)$. We now present the following result.
\begin{theorem}\label{the4}
	The power allocation profile
	\begin{align}
	\beta&=\frac{\varphi_1^*+\varphi_2^*+2\varphi_1^*\varphi_2^*}{2(1-\mu)\varphi_1^*\varphi_2^*}+\nonumber\\
	&\ \ \  -\frac{\sqrt{(\varphi_1^*+\varphi_2^*)^2+4\varphi_1^*\varphi_2^*(\varphi_1^*+\mu\varphi_2^*+\mu\varphi_1^*\varphi_2^*)}}{2(1-\mu)\varphi_1^*\varphi_2^*}\label{bet1}	
	\end{align}
    guarantees the same performance in terms of transmission rate for both UEs.
\end{theorem}
\begin{proof}
	See Appendix~\ref{App_D}. \phantom\qedhere	
\end{proof}
\vspace{-2mm}

The reachable SIR threshold can be easily calculated by substituting \eqref{bet1} into any of $\gamma_1$ or $\gamma_2$ expressions, which yields
\begin{align}
\gamma_1&=\gamma_2=-\frac{\varphi_1^*+\varphi_2^*}{2(\varphi_1^*+\mu\varphi_2^*(1+\varphi_1^*))}+\nonumber\\
&\ \ \ +\frac{\sqrt{(\varphi_1^*+\varphi_2^*)^2+4\varphi_1^*\varphi_2^*(\varphi_1^*+\mu(1+\varphi_1^*)\varphi_2^*)}}{2(\varphi_1^*+\mu\varphi_2^*(1+\varphi_1^*))}.\label{gam}
\end{align}

\begin{corollary}\label{cor2}
	For equal-rate allocation, the optimum decoding order for the UEs is $\varphi_2^*\ge \varphi_1^*$. 	
\end{corollary}
\vspace{-4mm}
\begin{proof}
See Appendix~\ref{App_E}. \phantom\qedhere	
\end{proof}
Let's say $\mathrm{UE}_A$ and $\mathrm{UE}_B$ are about to be served, then, Corollary~\ref{cor2}'s result states how they should be ordered at the BS for optimum performance. This is, one should calculate ${\varphi}^*$ for each $\mathrm{UE}$ according to \eqref{app} and then if ${\varphi}^*_A>{\varphi}^*_B\rightarrow  \mathrm{UE}_1=\mathrm{UE}_B,\ \mathrm{UE}_2=\mathrm{UE}_A$, otherwise if ${\varphi}^*_A<{\varphi}^*_B\rightarrow \mathrm{UE}_1=\mathrm{UE}_A,\ \mathrm{UE}_2=\mathrm{UE}_B$.

Fig.~\ref{Fig3}a shows the ratio $\varphi_1^*/\varphi_2^*$ for maximum transmission rate as a function of $\varphi_1^*+\varphi_2^*$. Since $\varphi_2^*\ge\varphi_1^*$, the ratio is always smaller or equal than 1. Notice that as $\varphi_1^*+\varphi_2^*$ increases and/or $\mu$ decreases, it is more desirable scheduling together for sharing the same spectrum resources those  $\mathrm{UE}$s having contrasting values of $\varphi_i^*$. According to Theorem~\ref{the2}, the way to control this lies on checking the $\mathrm{UE}$s reliability constraints and the ratio  between average signal and interference power. The more heterogeneous (homogeneous) the product of these parameters, the more (less) contrasting the values of $\varphi_i^*$. Additionally, according to the figure when SIC always fails the best performance comes when $\varphi_1^*=\varphi_2^*$.

A comparison with the OMA counterpart is required, especially because under OMA setup, both $\mathrm{UEs}$ require orthogonal spectrum resources, either in time or frequency. Although optimum OMA design usually conduces to different partition of resources, this is not practically viable since  requires cumbersome multi-user synchronization and scheduling tasks for partitioning the time, and could be even infeasible when bandwidth partition is required. Therefore, herein we assume equal time/frequency partition. Since $\varphi_i^*$ represents the SIR threshold required for $\mathrm{UE}_i$ for satisfying its reliability constraint if it is operating alone in the channel, its maximum allocable transmit rate is $\frac{1}{2}\log_2(1+\varphi_i^*)$, hence it is not the same for both $\mathrm{UEs}$. The goal in this subsection lies in achieving equal-rate performance, thus, the rate under the OMA setup is determined by the worst $\mathrm{UE}$ performance as $\frac{1}{2}\log_2(1+\min(\varphi_1^*,\varphi_2^*))=\frac{1}{2}\log_2(1+\varphi_1^*)$. Therefore, both $\mathrm{UEs}$ may be allocated with the same transmission rate, but while $\mathrm{UE}_1$ performs with $\epsilon_1$, $\mathrm{UE}_2$ performs with an outage performance smaller than $\epsilon_2$. Hence, we present the following result.
\begin{figure}[t!]
	\includegraphics[width=0.47\textwidth,right]{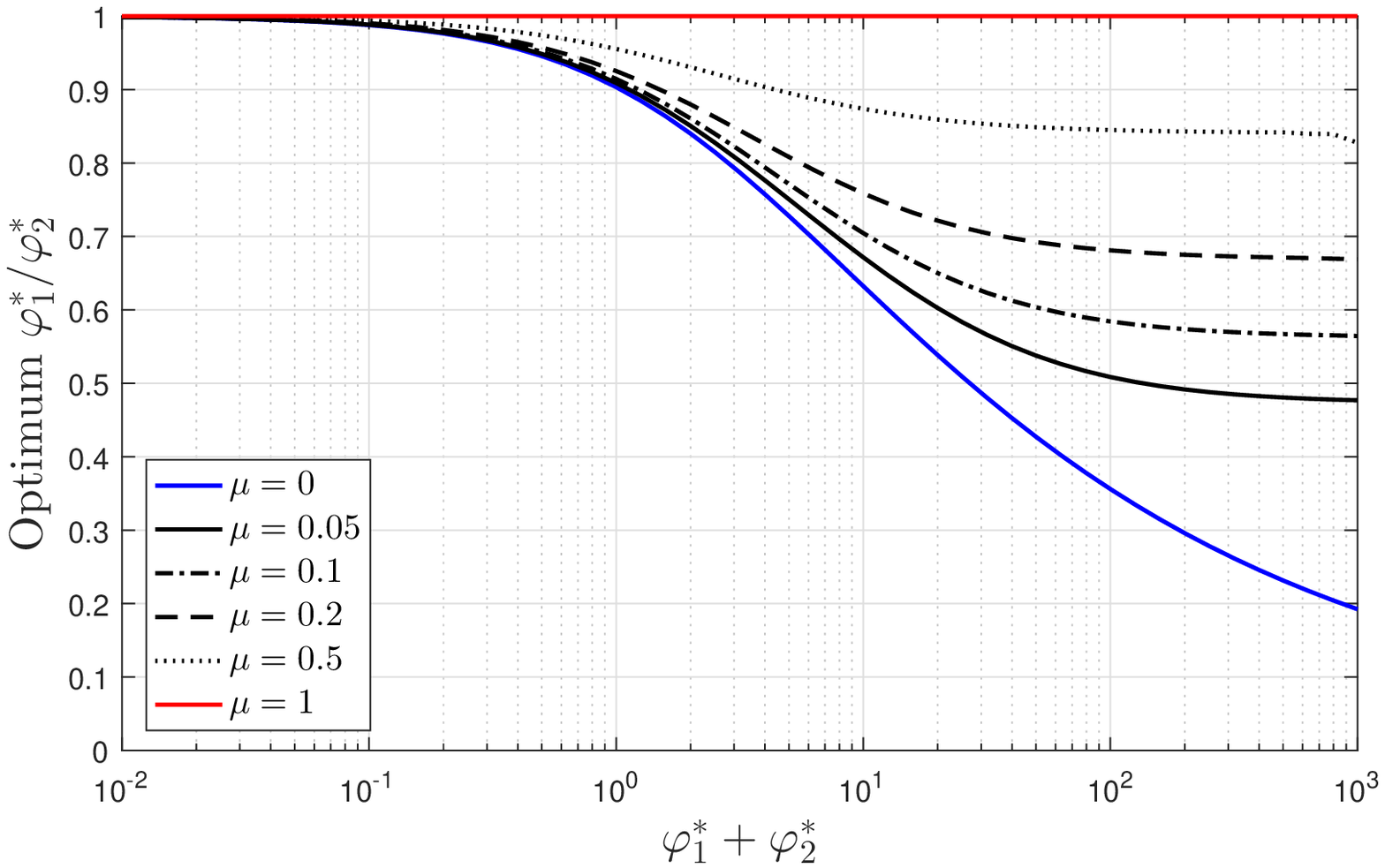}\vspace{2mm}\\
	\includegraphics[width=0.47\textwidth,right]{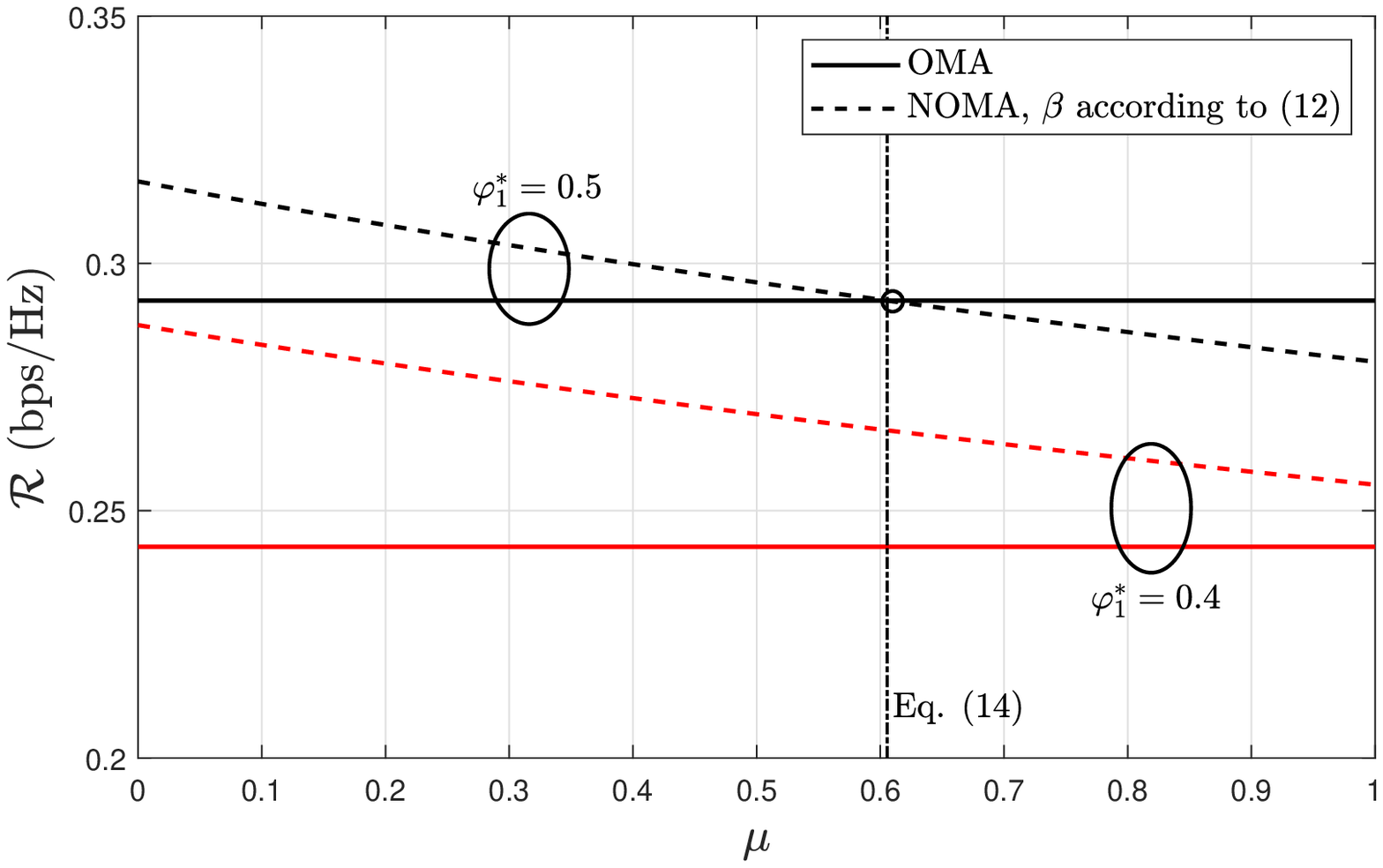}
	\caption{(a) Optimum ratio $\varphi_1^*/\varphi_2^*$ as a function of $\varphi_1^*+\varphi_2^*$ for $\mu\in\{0,0.05,0.1,0.2,0.5,1\}$ (top); (b) Transmission rate in bps/Hz as a function of $\mu$, for NOMA and OMA, $\varphi_2^*=0.6$ and $\varphi_1^*\in\{0.4,0.5\}$ (bottom).}		
	\label{Fig3}
\end{figure}

\begin{corollary}\label{re2}
	For equal-rate allocation and $M=2$, NOMA outperforms OMA when
\begin{align}
\mu<\frac{(\varphi_2^*-\varphi_1^*)(1+\sqrt{1+\varphi_1^*})}{\varphi_1^*\varphi_2^*\sqrt{1+\varphi_1^*}}.\label{muE}
\end{align}
\end{corollary}
\begin{proof}
	For the equal-rate allocation, the NOMA scheme overcomes the OMA configuration when
	\begin{align}
	\log_2(1+\gamma_1)&>\frac{1}{2}\log_2(1+\varphi_1^*)\nonumber\\
	\gamma_1&>\sqrt{1+\varphi_1^*}-1.\label{gam1}
	\end{align}
	Now, substituting \eqref{gam} into \eqref{gam1}, and after some algebraic manipulations we attain the condition given in \eqref{muE}.	
\end{proof}

Based on \eqref{muE} notice that NOMA performance probably overcomes OMA's if $\varphi_1^*$ and $\varphi_2^*$ are sufficiently different values. On the other hand, when $\varphi_2^*=\varphi_1^*$ OMA will perform better. Therefore, when scheduling users to work on the same channel with NOMA they should have sufficiently different topological characteristics, e.g., different $\frac{r_i^{-\alpha}}{\sum_{j\in\Phi\backslash\{b_0\}}r_{j,i}^{-\alpha}}$, and/or QoS requirements, e.g., different $\epsilon_i$, as shown in \eqref{app}. This is corroborated in Fig.~\ref{Fig3}b where we can notice that the gap between both schemes increases when $\varphi_2^*-\varphi_1^*$ also increases. Fig.~\ref{Fig3}b also corroborates Corollary~\ref{re2} and notice that when evaluating \eqref{muE} with $\varphi_1^*=0.4,\ \varphi_2=0.6$ we obtain that the required values of $\mu$ for which NOMA performs worse than OMA are $\mu\ge 1.53$, which are infeasible since $\mu\le 1$, thus, for that setup NOMA will always be the better choice.
\subsection{Maximum Sum-rate allocation}\label{power_s}
The previous subsection addressed the problem of finding the power allocation profile and scheduling order such that both UEs can operate with the same maximum possible rate. This is equivalent to find $\max\limits_{\beta} \min(\gamma_1,\gamma_2)$, since increasing $\gamma_1$ conduces always to decrease $\gamma_2$ and vice versa. Herein we are going to focus on finding the  allocation strategy for maximum sum-rate, thus, $\max\limits_{\beta}\mathcal{R}$, where
\begin{align}
\mathcal{R}=R_1+R_2&=\log_2(1+\gamma_1)+\log_2(1+\gamma_2)\nonumber\\
&=\log_2\Big((1+\gamma_1)(1+\gamma_2)\Big)=\log_2\tilde{\gamma},
\end{align}
where $\tilde{\gamma}=(1+\gamma_1)(1+\gamma_2)$, thus, $\max\limits_{\beta}\mathcal{R}=\max\limits_{\beta}\tilde{\gamma}$. 

Notice that solving that problem could be extremely difficult, if not impossible, due to the tangled dependence of $\tilde{\gamma}$ on $\beta$, as shown next
\begin{align}
\tilde{\gamma}&=(1+\gamma_1)(1+\gamma_2)\nonumber\\
&=\frac{(1+\varphi_1^*)(1+\varphi_2-\beta\varphi_2^*(1-\mu))}{(1+(1-\beta)\varphi_1^*)(1+\beta\mu\varphi_2^*)}.\label{tildeG}
\end{align}
Only for $\mu=0$ the problem solution comes easy as follows
\begin{align}
\tilde{\gamma}|_{\mu=0}&=(1+\varphi_1^*)\frac{1+\varphi_2^*(1-\beta)}{1+\varphi_1^*(1-\beta)},\\
\frac{d\tilde{\gamma}|_{\mu=0}}{d\beta}&=(1+\varphi_1^*)\frac{\varphi_1^*-\varphi_2^*}{(1+\varphi_1^*(1-\beta))^2},
\end{align}
thus, $\tilde{\gamma}|_{\mu=0}$ is a strictly decreasing (increasing) function for $\varphi_1^*<\varphi_2^*$ ($\varphi_1^*>\varphi_2^*$) and any value of $\beta\in[0,1]$.
\begin{remark}
	Therefore, if $\mu=0$ then the optimal performance for the NOMA setup, in terms of maximum sum-rate, is obtained by setting
	\begin{equation}
	\beta^*= \left\{ \begin{array}{lll}	
	\!1/2,& \!\mathrm{if}\  \varphi_2^*\ge\varphi_1^*&\!\rightarrow \tilde{\gamma}|_{\mu=0}=\frac{(1+\varphi_1^*)(2+\varphi_2^*)}{2+\varphi_1^*}\\
	\!1,\!&\! \mathrm{if}\ \varphi_2^*<\varphi_1^* &\!\rightarrow \tilde{\gamma}|_{\mu=0}=1+\varphi_1^*
	\end{array}
	\right.\!\!\!\!.\label{bmu0}
	\end{equation}	
	Of course, it is desirable having/setting $\varphi_2^*>\varphi_1^*$ while using $\beta=1/2$, otherwise the optimum performance  is when $\beta=1$ which matches the OMA setup when only one user is being served at the time. Also, notice that $\beta^*=1/2$ conduces to a greater value of $\tilde{\gamma}|_{\mu=0}$ than when $\beta=1$.
\end{remark}

For the general case when $\mu>0$ we propose using the following result henceforth.
\begin{proposition}\label{pro1}
	With $\bar{\gamma}=\gamma_1+\gamma_2$, $\tilde{\gamma}$ approximates to
	\begin{align}
	\tilde{\gamma}\approx \Big[1+\frac{1}{2}\bar{\gamma}\Big]^2,\label{gamB}
	\end{align}
\end{proposition}
\begin{proof}
	See Appendix~\ref{App_F}. \phantom\qedhere
\end{proof}

Finally, we provide two approximate results characterizing the optimum performance of NOMA and the region for which NOMA outperforms OMA for the general case of $\mu\ge 0$.
\begin{theorem}\label{the5}
	The optimal power allocation profile for maximum sum-rate in the NOMA setup is
	\begin{equation}
	\beta^*= \left\{ \begin{array}{ll}	
	1/2,&\ \mathrm{if}\  \mu<\frac{2(\varphi_2^*-\varphi_1^*)+\varphi_1^*(\varphi_2^*-2)}{\varphi_1^*\varphi_2^*(1+\varphi_1^*)}\\
	1,& \ \mathrm{otherwise}
	\end{array}
	\right..\label{eqBo}
	\end{equation}
\end{theorem}
\begin{proof}
	See Appendix~\ref{App_G}. \phantom\qedhere
\end{proof}
Notice that when $\beta=1/2$ we have
\begin{align}
\bar{\gamma}=\frac{\varphi_1^*}{2+\varphi_1^*}+\frac{\varphi_2^*}{2+\mu\varphi_2^*}.\label{gamB2},
\end{align}
and for $\beta=1\rightarrow \bar{\gamma}=\varphi_1^*$.

\begin{figure}[t!]
	\includegraphics[width=0.47\textwidth,right]{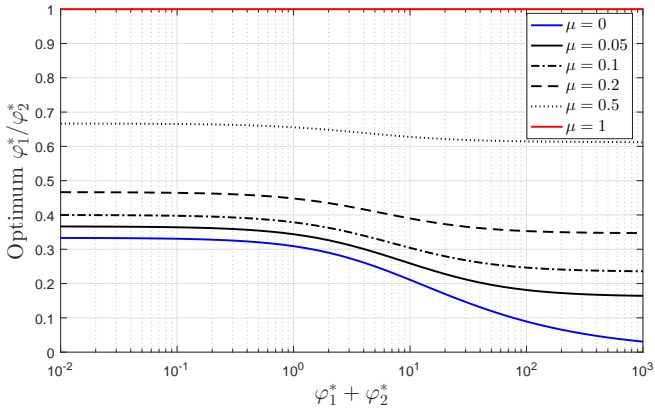}
	\caption{Optimum ratio $\varphi_1^*/\varphi_2^*$ as a function of $\varphi_1^*+\varphi_2^*$ for $\mu\in\{0,0.05,0.1,0.2,0.5,1\}$ and $\beta=1/2$.}		
	\label{Fig4}
\end{figure}
As done in the previous subsection for the optimum equal-rate allocation, herein we show in Fig.~\ref{Fig4} the optimum ratio $\varphi_1^*/\varphi_2^*$ as a function of $\varphi_1^*+\varphi_2^*$ for the maximum sum-rate allocation problem. Again, the trend is that the larger $\varphi_1^*+\varphi_2^*$ and/or smaller $\mu$, the more contrasting should be the values of $\varphi_i^*$ for optimum performance. However, in this case the dependence is much weaker on $\varphi_1^*+\varphi_2^*$ than for the case of equal-rate allocation which may facilitate the design and scheduling tasks in a practical system.

As in the previous subsection, herein we also consider OMA with equal partition of spectrum resources, either in time or frequency, thus, the maximum transmission rate for each $\mathrm{UE}_i$ signal is $\frac{1}{2}\log_2(1+\varphi_i^*)$. Then, the following result holds.
\begin{theorem}\label{the6}
	The NOMA setup with $M=2$ outperforms  always OMA in terms of maximum sum-rate when $\mu=0$, while \textit{almost surely}\footnote{For $\mu>0$ the proof takes advantage of result in Proposition~\ref{pro1} since using the exact expression of $\tilde{\gamma}$ is cumbersome. Therefore, \eqref{cor4} is expected to hold given the accuracy of \eqref{gamB}, hence the term \textit{almost surely}.} when 
	\begin{align}
	\mu&<\frac{\sqrt{2}(\varphi_1^*+2)^2\sqrt{2+\varphi_1^*+\varphi_2^*}}{\varphi_1^{*2}(2\varphi_1^*+3)+2\varphi_2^*(\varphi_1^*+2)^2}+\nonumber\\
	&\qquad\ -\frac{4\varphi_1^{*3}+6\varphi_1^{*2}+\varphi_1^{*2}\varphi_2^*+6\varphi_1^*\varphi_2^*+8\varphi_2^*}{\varphi_2^*(\varphi_1^{*2}(2\varphi_1^*+3)+2\varphi_2^*(\varphi_1^*+2)^2)}\label{cor4}
	\end{align}
	with $\varphi_2^*>\varphi_1^*$.
\end{theorem}
\begin{proof}
See Appendix~\ref{App_H}. \phantom\qedhere	
\end{proof}
\begin{figure}[t!]
	\includegraphics[width=0.47\textwidth,right]{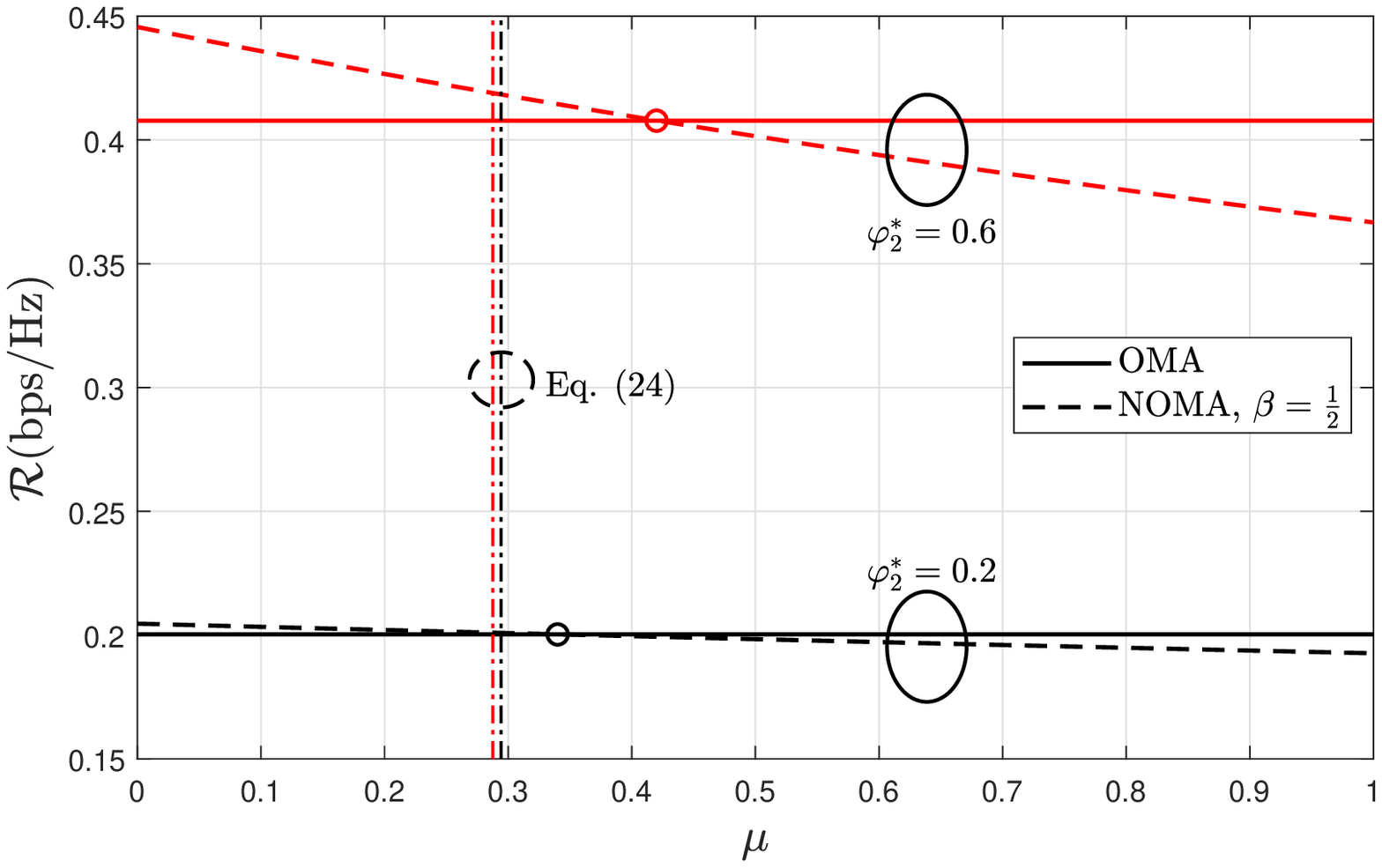}\vspace{2mm}\\ \includegraphics[width=0.47\textwidth,right]{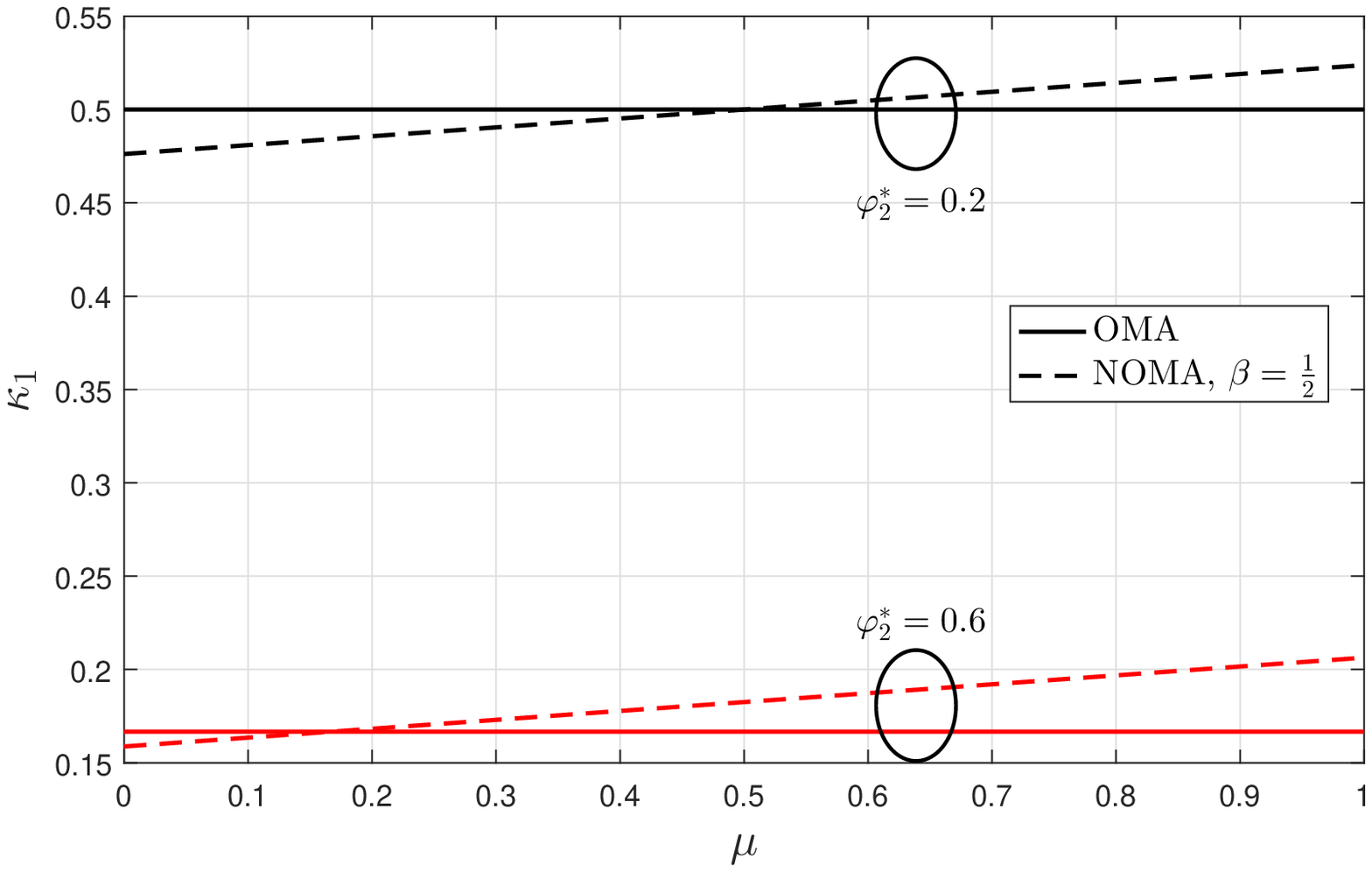}
	\caption{(a) Transmission rate in bps/Hz (top) and (b) Fairness coefficient (bottom), as a function of $\mu$, for NOMA and OMA, $\varphi_1^*=0.1$ and $\varphi_2^*\in\{0.2,0.6\}$.}		
	\label{Fig5}
\end{figure}
See Fig.~\ref{Fig5}a for a comparison between NOMA and OMA schemes and notice that \eqref{cor4} holds in each case. Also, and as shown in the figure, we can assure that NOMA always overcomes OMA when $\mu\rightarrow 0$, as it was claimed by Theorem~\ref{the6}. Finally, a comparison in terms of fairness when allocating the transmission rate is required. While the power allocation scheme proposed in the previous subsection allows reaching maximum fairness, this is not longer the case when the goal is maximizing the sum-rate. As a measure of fairness let us consider the ratio between the SIR thresholds of both $\mathrm{UE}$s, since this is analytically easier to handle than the transmission rate itself. We define the fairness coefficient of $\mathrm{UE}_i$ with respect to $\mathrm{UE}_j$, with $i\ne j$ and $i,j\in\{1,2\}$, as $\kappa_i^{\mathrm{oma}}=\varphi^*_i/\varphi^*_{j}$ and $\kappa_i^{\mathrm{noma}}=\gamma_i/\gamma_j$, and the closer $\kappa_i$ is from $1$, the fairer the rate allocation. For optimality we set $\varphi_1^*\le \varphi_2^*$ according to Theorem~\ref{the6}, then, focusing on $\kappa_1$ we have that $\kappa_1^{\mathrm{oma}}=\varphi_1^*/\varphi_2^*\le 1$ and
	\begin{align}
	\kappa_1^{\mathrm{noma}}\stackrel{(a)}{=}\frac{\frac{\frac{1}{2}\varphi_1^*}{1+\frac{1}{2}\varphi_1^*}}{\frac{\frac{1}{2}\varphi_2^*}{1+\frac{1}{2}\mu\varphi_2^*}}=\frac{\varphi_1^*(2+\mu\varphi_2^*)}{\varphi_2^*(2+\varphi_1^*)}=\kappa_1^{\mathrm{oma}}\frac{2+\mu\varphi_2^*}{2+\varphi_1^*},\label{k1}
	\end{align}	
	where $(a)$ comes from setting $\beta=1/2$ in \eqref{gami} as this was proved to be the optimum when NOMA provides greater sum-rate than OMA. Notice that as long as $\mu<\varphi_1^* /\varphi_2^*$, OMA is fairer as also illustrated in Fig.~\ref{Fig5}b, however this is not advantageous when heterogeneous services need to be provided. In practice, when considering the entire set of frequency-time resources, previous results are useful when making the pairing of those $\mathrm{UEs}$ that will be scheduled with the same spectrum resources.
\section{Algorithmic and Performance Remarks}\label{finalre}
Section~\ref{FB} focused on the optimum allocation of the transmission rates of concurrent downlink NOMA users under practical reliability constraints and given their preassigned transmit powers. Meanwhile, Section~\ref{power} explained how to find these transmit powers for optimum performance in the 2-user NOMA scenario. Next, we summarize how to integrate  both rate and power allocation schemes for attaining the optimum performance. This is discussed in Subsection~\ref{s1}, while Subsection~\ref{s2} shows how to calculate the average rate performance for NOMA and OMA setups as it will be used for illustrating numerical results in Section~\ref{results}.
\subsection{Optimum Rate-Power Allocation}\label{s1}
By combining the results from Section~\ref{FB} and~ \ref{power}, the optimum rate-power allocation algorithm for a 2-user downlink NOMA setup without instantaneous CSI at the BS is given as:
\begin{enumerate}
	\item Compute ${\varphi}^*$ according to \eqref{app} for each $\mathrm{UE}$.
	\item Set the $\mathrm{UE}$ with the lowest ${\varphi}^*$ as first in the decoding order, while the other becomes second. This is grounded on Corollary~\ref{cor2} and Theorem~\ref{the6} results.
	\item Compute $P_1=\beta P_T$ and $P_2=(1-\beta)P_T$, where $\beta$ is given by \eqref{the4} and \eqref{eqBo} for equal-rate and maximum sum-rate allocation, respectively.
	\item Compute $\gamma_i,\ i=1,2$ according to \eqref{gami} with $M=2$. For the case of equal-rate allocation it is also possible using \eqref{gam}, which yields the same result.
	\item Set $P_i$ and $\log_2(1+\gamma_i)$ as transmit power and rate, respectively, of the signal intended to $\mathrm{UE}_i$.
\end{enumerate}
\subsection{Average performance of NOMA and OMA}\label{s2}
For each $\mathrm{UE}_i$,  the distribution of its SIR threshold $\gamma_i$ with respect to the PPP spatial randomness was found in Subsection~\ref{dSIR} for a fixed power allocation profile. When BSs allocate their transmit power based on the network deployment, as it is the case of the power control scheme for 2-user NOMA proposed in Section~\ref{power} for equal-rate and maximum sum-rate problems, expressions \eqref{eqth3} and \eqref{eqth4} in Subsection~\ref{dSIR} do not hold anymore. This is because $P_i$ is now another random variable in the PPP. In such cases characterizing the distribution of $\gamma_i$, hence, the distribution of the allocated rate, seems analytically intractable because of the tangled dependence of $\beta$ on $\varphi_i^*$ according to \eqref{bet1}, and \eqref{eqBo} for $\mu>0$, respectively. Only for the case of perfect SIC in a system where the goal is maximizing the sum-rate, such distribution can be obtained since the optimum $\beta$ is deterministically $1/2$ according to \eqref{eqBo}, thus, \eqref{eqth3} holds and \eqref{ztheta} becomes
	\begin{align}
	z_i(\theta)&=\left(\frac{P_T}{2\theta}-\frac{\mu P_T}{2}\sum\limits_{j=1}^{i -1} 1 -\frac{P_T}{2}\sum\limits_{j=i+1}^{2} 1 \right) \frac{\epsilon_i}{P_T}\nonumber\\
	&= \frac{\epsilon_i}{2}\Big(\frac{1}{\theta}+(1-\mu)(i-1)-1\Big).\label{zi}
	\vspace*{-8mm}
	\end{align}

For a comparison between OMA and NOMA in terms of average rate performance, $\bar{\mathcal{R}}$, over the PPP, following expressions hold
\begin{itemize}
	\item  equal-rate allocation
	\begin{align}
	\bar{\mathcal{R}}\!=\!\left\{ \begin{array}{ll}	
	\!\!\!\mathbb{E}\Big[\frac{1}{2}\log_2\!\Big(\!1\!+\!\min\big(\varphi_1^*,\varphi_2^*\big)\!\Big)\Big|\Phi\Big]\!\!:& \! \mathrm{OMA}\\
	\!\!\!\mathbb{E}\Big[\log_2\!\big(1+\gamma_1\big)\Big|\Phi\Big]\!\!:& \! \mathrm{NOMA}
	\end{array}\!\!\!,
	\right.
	\end{align}
	\item  maximum sum-rate allocation
	\begin{align}
	\bar{\mathcal{R}}\!=\!\left\{ \begin{array}{ll}	
	\!\!\!\mathbb{E}\Big[\frac{1}{2}\log_2\!\big(1\!+\!\varphi_1^*\big)\big(1+\varphi_2^*\big)\Big|\Phi\Big]\!\!:&  \mathrm{OMA}\\
	\!\!\!\mathbb{E}\Big[\log_2\!\tilde{\gamma}\Big|\Phi\Big]\!\!:&  \mathrm{NOMA}
	\end{array}\!\!\!,
	\right.
	\end{align}
	where $\tilde{\gamma}$ is given in \eqref{tildeG}.
\end{itemize}
These expressions are used in Section~\ref{results} for numerical analysis. Notice also that as a direct consequence of the discussions in Subsection~\ref{dSIR}, the average transmission rate performance of NOMA and OMA setups does not depend on the network density when randomizing the network topology through the PPP. This is corroborated in Fig.~\ref{Fig6} where the average is taken over the network realizations. The values of the system parameters are those given at the beginning of the next section.
	\begin{figure}[t!]
		\includegraphics[width=0.47\textwidth,right]{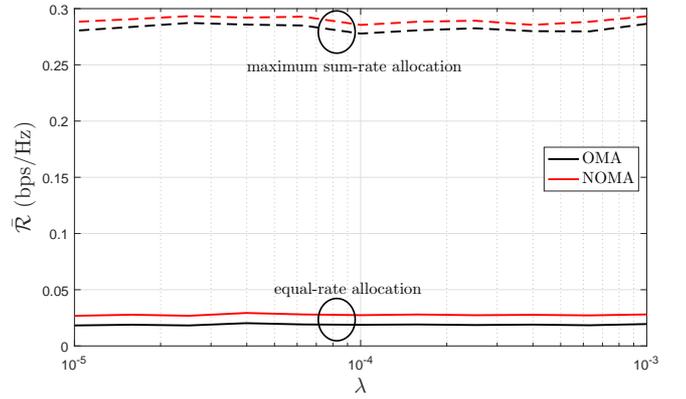}
		\caption{Average rate for OMA and NOMA setups operating with equal-rate and maximum sum-rate configurations and $\epsilon_i=10^{-2}$.}		
		\label{Fig6}
	\end{figure}
\section{Numerical Results}\label{results}
Herein we compare numerically the performance of NOMA and OMA setups under the proposed distributed rate control methodology. We simulate completely random scenarios by generating 50000 instances (Monte Carlo runs) of a PPP  and a sufficiently large area such that 1000 BSs are placed on average. Since the average performance of NOMA and OMA setups do not depend on $\lambda$, we are going to use a generic value of $\lambda=10^{-4}/\mathrm{m}^2$ $(100/\mathrm{km}^2)$. Unless stated otherwise, results are obtained by setting  $\alpha=4$ and $\mu=0.1$. We utilize the optimum values of $\beta$ given in Section~\ref{power_e} and Section~\ref{power_s} for equal-rate and maximum sum-rate allocations, respectively. 
%
\begin{table*}[!t]
	\centering
	\caption{Average Transmission Rate for Benchmark OMA and NOMA schemes.}
	\begin{tabular}{c|c|c}
		\toprule
		\textbf{Scheme} & \textbf{Equal-rate} & \textbf{Maximum sum-rate} \\ \bottomrule
		OMA  & $\mathbb{E}\Big[\frac{1}{2}\log_2\big(1+\min(h_1r_1^{-\alpha} P_T/I_1,h_2r_2^{-\alpha} P_T/I_2)\big)\Big|\Phi\Big]$ & $\mathbb{E}\Big[\frac{1}{2}\log_2(1\!+\!h_1r_1^{-\alpha} P_T/I_1)(1\!+\!h_2r_2^{-\alpha} P_T/I_2)\Big|\Phi\Big]$ \\[8pt] 
		NOMA  & $\mathbb{E}\Big[\log_2(1\!+\!\mathrm{SIR}_1)$ with $1/2\!\le \!\beta\!\le\! 1\!:$ $\mathrm{SIR}_1\!=\!\mathrm{SIR}_2\Big|\Phi\Big]$ & $\mathbb{E}\Big[\max_{1/2\le \beta\le 1}\log_2(1+\mathrm{SIR}_1)(1+\mathrm{SIR}_2)\Big|\Phi\Big]$ \\
		\bottomrule
	\end{tabular}\label{table2}
\end{table*}

Additionally, the performance  when perfect CSI and knowledge of instantaneous levels of interference at the $\mathrm{UE}s$ are fully available at the BS side, is illustrated as benchmark. As in \cite{Zhang.2016}, the decoding order is now established for holding $h_1r_1^{-\alpha}/I_1<h_2r_2^{-\alpha}/I_2$. Since $h_i$ and $I_i$ are deterministic values in such scenario, the rate can be allocated directly from \eqref{eq4} for each channel realization as $2^{\mathrm{SIR}_i}-1$. According to Shannon, this rate allows operating with an arbitrarily small error probability as the blocklength goes to infinity \cite{Polyanskiy.2010}\footnote{The trade-off between the blocklength and the error probability was analytically investigated in \cite{Polyanskiy.2010} but herein we assume the asymptotic scenario for which the blocklength is infinity, thus $\epsilon_i\rightarrow 0$ if the transmission rate is set to be $2^{\mathrm{SIR}_i}-1$ for each channel realization.}.  Finally, the average allocated transmission rates under the benchmark NOMA scheme, and also under its OMA counterpart, appear shown in Table~\ref{table2} for equal-rate and maximum sum-rate problems.
\begin{figure}[t!]
	\includegraphics[width=0.47\textwidth,right]{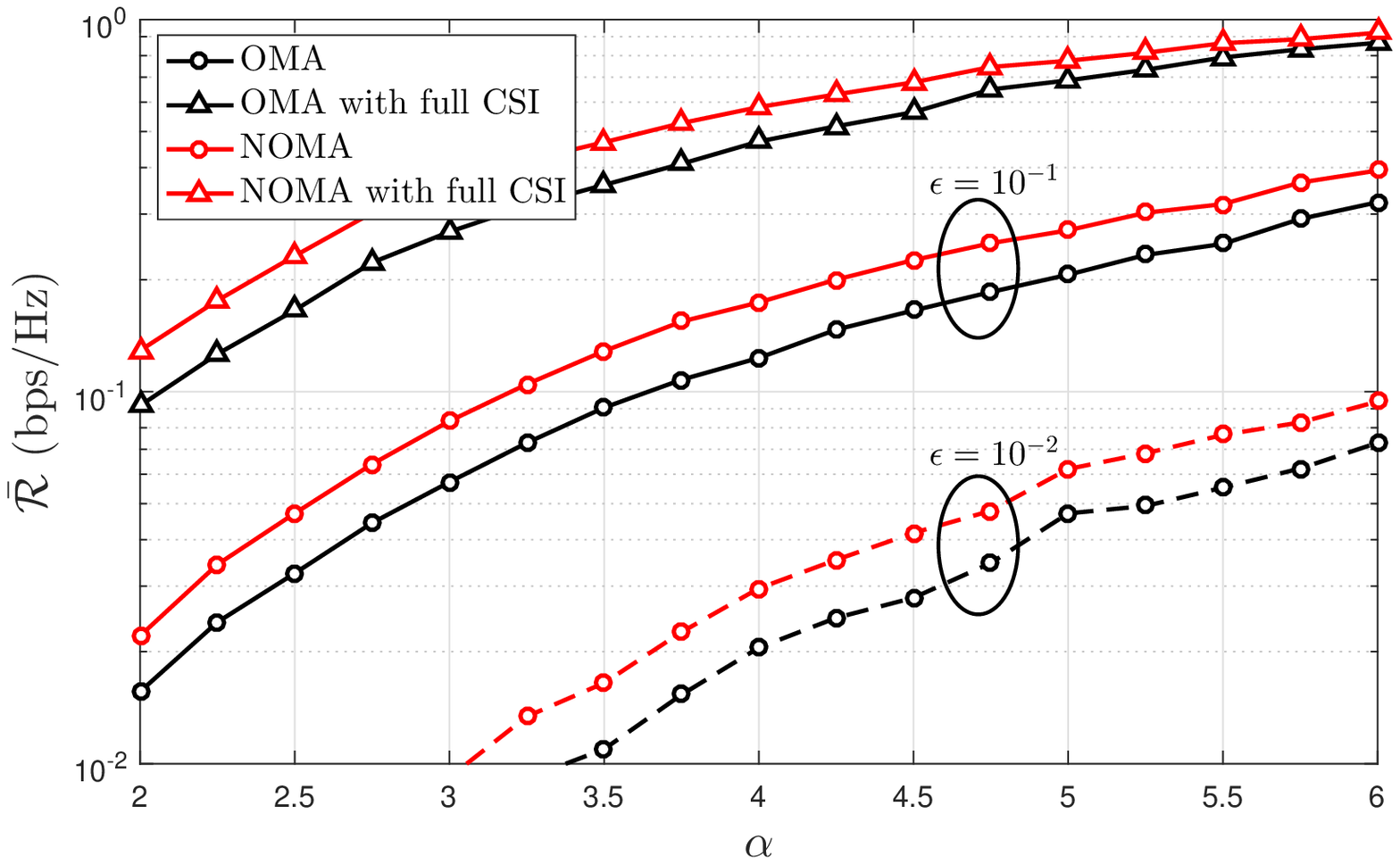}\vspace{2mm}\\ 
	\includegraphics[width=0.47\textwidth,right]{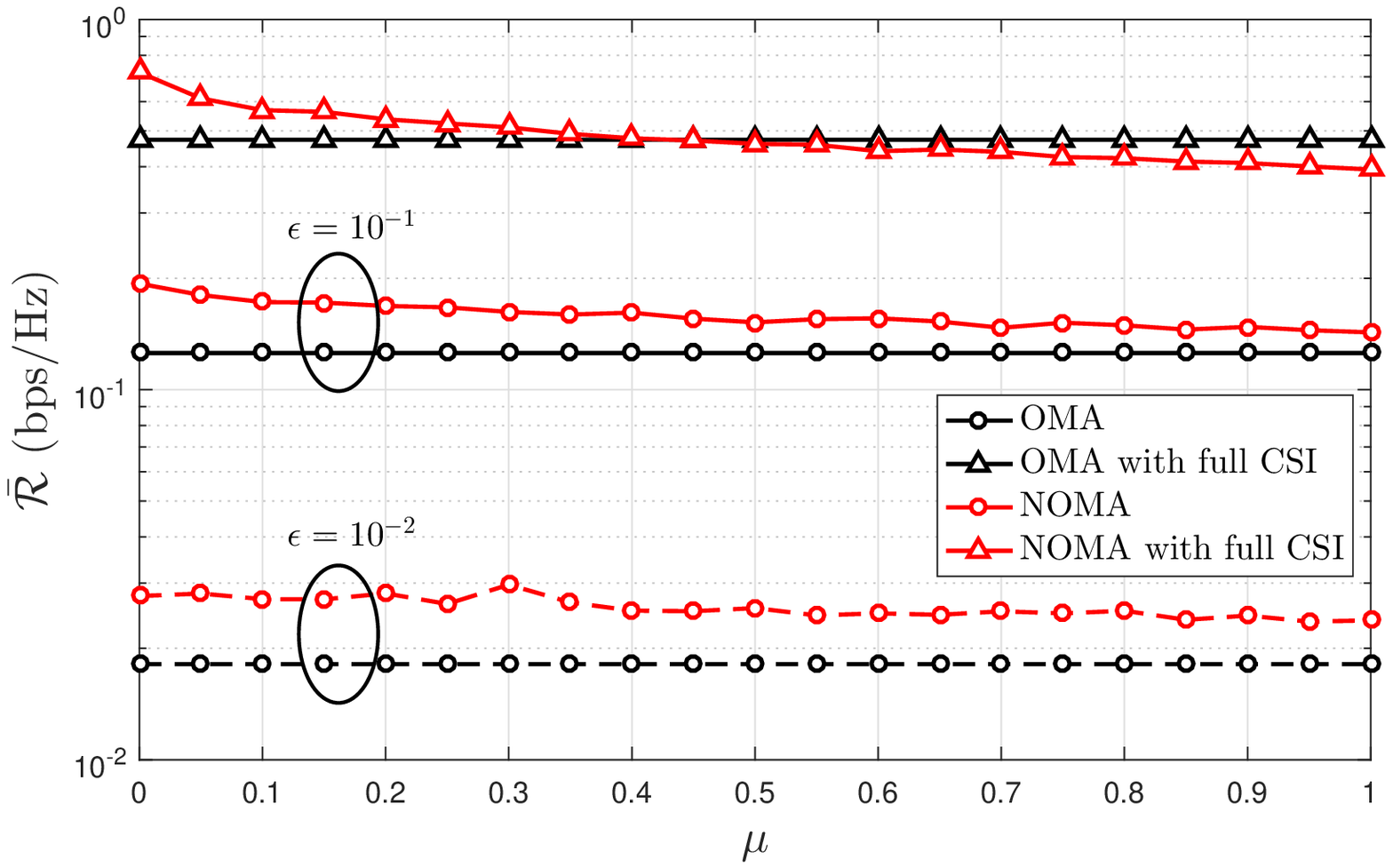}
	\caption{Average rate for OMA and NOMA setups as a function of (a) $\alpha$ (top), (b) $\mu$ (bottom), for equal-rate allocation and $\epsilon_i\in\{10^{-1},10^{-2}\}$.}
	\label{Fig7}
\end{figure}
\subsection{Equal-rate allocation}\label{subsub}
Fig.~\ref{Fig7} shows the average rate for OMA and NOMA setups when operating with equal-rate allocation, and as a function of $\alpha$ (Fig.~\ref{Fig7}a) and $\mu$ (Fig.~\ref{Fig7}b). 
Notice that for our CSI-free schemes the average performance of NOMA overcomes always the OMA's, and the gap increases for smaller SIC imperfection $\mu$ and/or path-loss exponent\footnote{This behavior can be corroborated by checking that indeed $\frac{d}{d\alpha}\big(\frac{\varphi_2^*}{\varphi_1^*}\big)<0$ in most of the PPP realizations for $\epsilon_1=\epsilon_2$. Notice that as $\frac{\varphi_2^*}{\varphi_1^*}$ decreases NOMA loses its advantages over OMA as discussed after Corollary~\ref{re2}.} $\alpha$.   
Meanwhile, NOMA under full CSI at the BS is more sensitive to $\mu$, and for $\mu>0.4$ its OMA counterpart performs better. These benchmark schemes provide an upper-bound performance not only because full CSI is considered at the BS, but also because only when transmitting over an infinite blocklength their results hold. Therefore and as expected, it is shown a large gap between these bounds and the performance of the OMA and NOMA schemes discussed in this work, and increases as the reliability constraint becomes more stringent. 
For all the schemes the performance improves when the path-loss exponent increases, which is expected according to our discussions in Subsection~\ref{dSIR}. In a nutshell, since interfering signals endure larger distances, as $\alpha$ increases their impact is affected more than the receive power of the desired signal itself, hence, larger values of $\mathrm{SIR}$ are more likely. 
Additionally, the performance of OMA does not depend on $\mu$, thus, it appears as a straight line in Fig.~\ref{Fig7}b,  and remarkably even in the worst case of $\mu=1$ NOMA attains a greater average performance than OMA when operating without CSI at the BS.
Notice that $\mu=1$ implies only that the interference coming from  UEs that come first in the decoding order cannot be canceled. Therefore, this affects only the denominator of \eqref{eq4}, and the UE being currently decoded has still chances of succeeding.
Obviously, the OMA setup could be optimized such that different partitions of time/frequency resources are allocated to each user, thus improving its performance and even overcoming NOMA. However, notice that this is extremely cumbersome since involves synchronization tasks if the split resource is time, and it could be even unfeasible if the resource is the frequency bandwidth. Meanwhile, NOMA only requires to use a different power level for the signal of each UE.
\begin{figure}[t!]
	\includegraphics[width=0.47\textwidth,right]{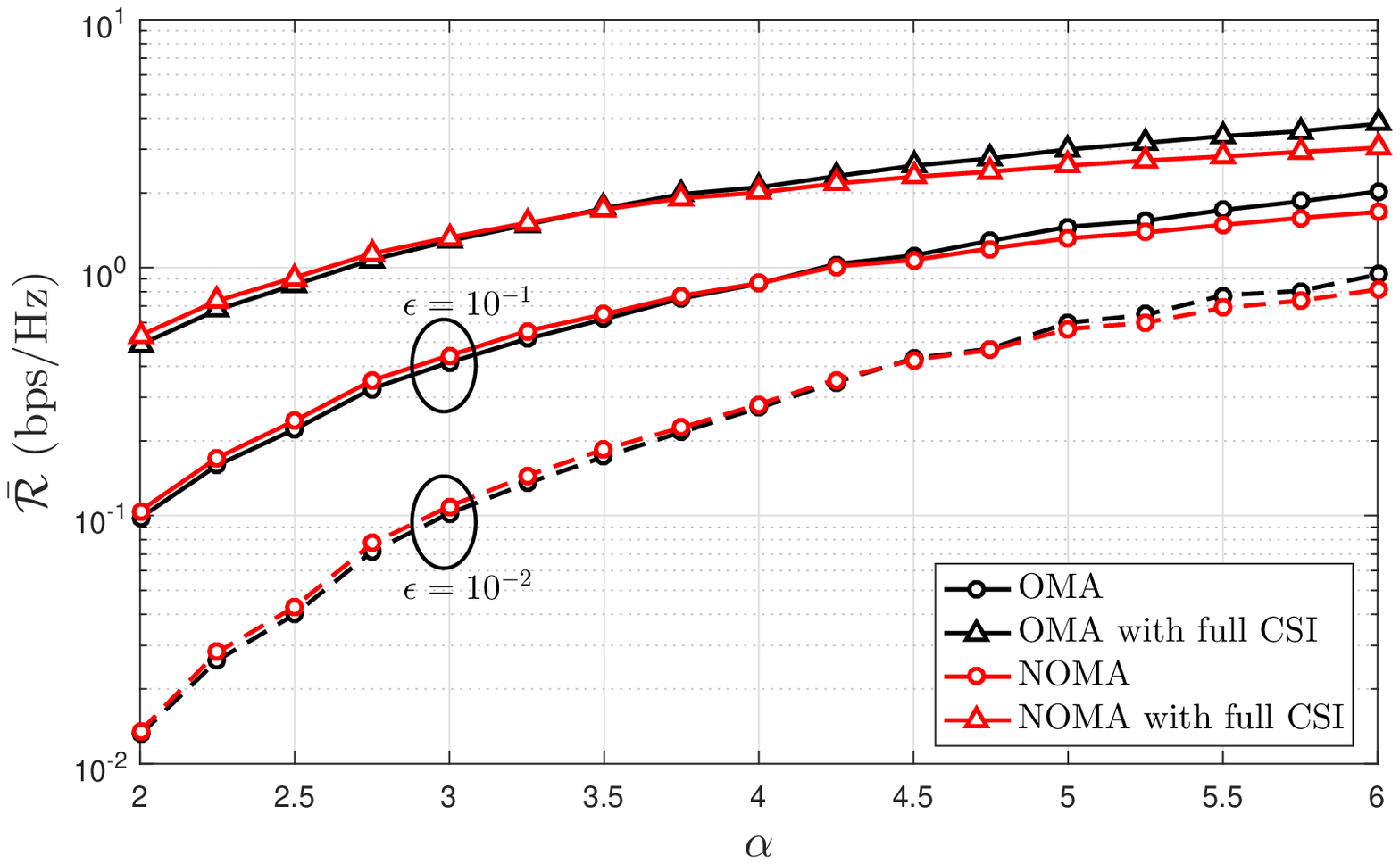}\vspace{2mm}\\
	\includegraphics[width=0.47\textwidth,right]{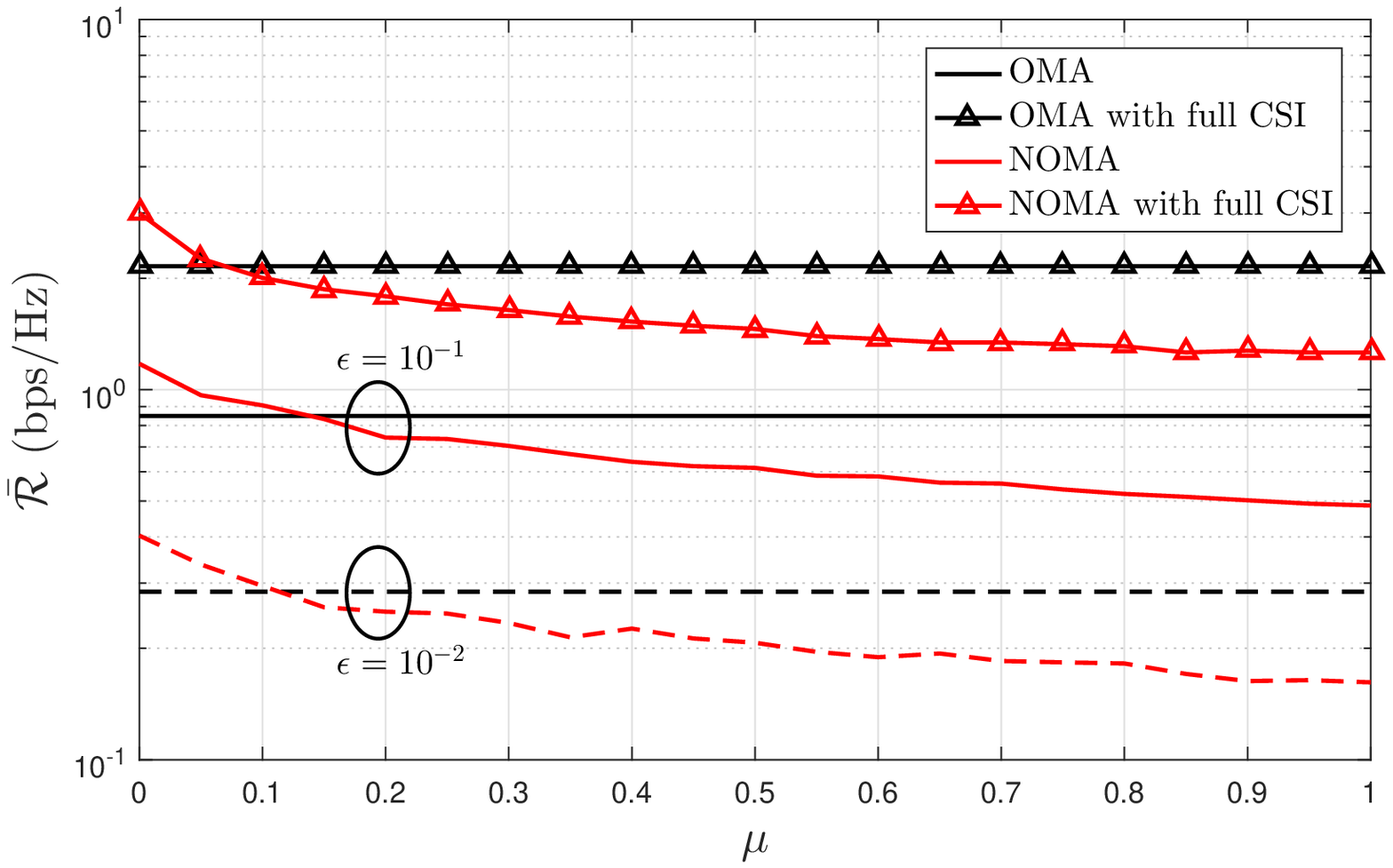}
	\caption{Average rate for OMA and NOMA setups as a function of (a) $\alpha$ (top), (b) $\mu$ (bottom), when operating with maximum sum-rate allocation and $\epsilon_i\in\{10^{-1},10^{-2}\}$.	}	
	\label{Fig8}
\end{figure}
\subsection{Maximum sum-rate allocation}
Fig.~\ref{Fig8} shows the average rate for OMA and NOMA setups when operating with maximum sum-rate allocation, and as a function of $\alpha$ (Fig.~\ref{Fig8}a) and $\mu$ (Fig.~\ref{Fig8}b). For this case the performance of OMA improves faster than the NOMA's when $\alpha$ increases, although when reliability constraint becomes more stringent both ideal CSI-free schemes tend to have the same improvement rate with respect to $\alpha$. Notice that it is required that $\mu\lesssim 0.1$ ($\lesssim 10\%$ of the co-interference is not canceled) so that NOMA is preferable to OMA. This is illustrated in Fig.~\ref{Fig8}b, where the gains of NOMA over OMA are of $39\%$, $38\%$ and $42\%$ when $\mu\rightarrow 0$ for the benchmark and CSI-free schemes with $\epsilon_i=10^{-1}$ and $\epsilon_i=10^{-2}$, respectively; while when $\mu\rightarrow 1$, OMA offers a performance gain around of $72\%$, $81\%$ and $71\%$. Therefore, failing to efficiently  eliminate the co-interference is critical for NOMA, and can be a  challenging issue when implementing it in practice. Finally and as commented in the Section~\ref{subsub}, benchmark schemes provide a loose upper-bound performance which is impossible to attain for any practical communication system.
\section{Conclusion}\label{conclusions}
In this paper, we proposed a rate allocation scheme for a downlink NOMA cellular system operating with reliability constraints.  
The allocated rate depends on: \textit{i)} the average receive power from the desired signal and from the interfering BSs; and \textit{ii)} the reliability constraint. This information is usually easy to get and therein lies the main advantage of the scheme. Additionally, we characterize analytically the distribution of the allocated SIR threshold in Poisson cellular networks in order to meet the link's reliability constraint. We derive and analyze the main conditions that are necessary so that NOMA overcomes the OMA alternative, while we discuss the optimum allocation strategies, e.g., rate and power allocation profile and optimum decoding order,  for the 2-UEs NOMA setup with equal-rate and maximum sum-rate goals. 
We compare numerically the performance of our allocation scheme with its ideal counterpart requiring full CSI at the BSs and infinitely long blocklength, and show how the gap increases as the reliability constraint becomes more stringent.
Also, results highlight the benefits of NOMA over OMA when the co-interference can be efficiently canceled, e.g., with an efficiency of $90\%$, specially when we aim to maximize the sum-rate.
\appendices 
\section{Proof of Lemma~\ref{the1}}\label{App_A}
	Setting the conditional link success probability to $1-\epsilon_i$, we have
\begin{align}
1-\epsilon_i&=\mathbb{P}\left(\frac{h_ir_i^{-\alpha}P_i}{h_ir_i^{-\alpha}\Big[\mu\sum\limits_{j=1}^{i-1}P_j+\sum\limits_{j=i+1}^{M}P_j\Big]+I_i}>\gamma_i\Bigg|\Phi\right)\nonumber\\
&=\mathbb{P}\left(h_i>\frac{I_ir_i^{\alpha}}{\frac{P_i}{\gamma_i}-\mu\sum\limits_{j=1}^{i-1}P_j-\sum\limits_{j=i+1}^{M}P_j}\Bigg|\Phi\right)\nonumber\\
&\stackrel{(a)}{=}\mathbb{E}\left[\mathbb{P}\left(h_i>\frac{r_i^{\alpha}P_T\sum\limits_{j\in\Phi\backslash \{b_0\}}g_jr_{j,i}^{-\alpha}}{\frac{P_i}{\gamma_i}-\mu\sum\limits_{j=1}^{i-1}P_j-\sum\limits_{j=i+1}^{M}P_j}\right)\Bigg|\Phi\right]\nonumber\\
&\stackrel{(b)}{=}\mathbb{E}\left[\exp\left(-\frac{r_i^{\alpha}P_T\sum\limits_{j\in\Phi\backslash \{b_0\}}g_jr_{j,i}^{-\alpha}}{\frac{P_i}{\gamma_i}-\mu\sum\limits_{j=1}^{i-1}P_j-\sum\limits_{j=i+1}^{M}P_j}\right)\Bigg|\Phi\right]\nonumber\\
&\stackrel{(c)}{=}\mathbb{E}\left[\prod_{j\in \Phi\backslash \{b_0\}}\exp\left(-\varphi_i r_i^{\alpha}g_jr_{j,i}^{-\alpha}\right)\Big|\Phi\right]\nonumber\\
&\stackrel{(d)}{=}\prod_{j\in \Phi\backslash \{b_0\}}\frac{1}{1+\varphi_i r_i^{\alpha}r_{j,i}^{-\alpha}},
\end{align}
where $(a)$ comes from using \eqref{Ii}, $(b)$ from using the complementary CDF (CCDF) of $h_i$, while $(c)$ from using the algebraic transformation $\exp\big({\sum_{i=1}^{n}k_i}\big)=\prod_{i=1}^{n}\exp({k_i})$ and setting
\begin{align}
\varphi_i=\frac{P_T}{\frac{P_i}{\gamma_i}-\mu\sum\limits_{j=1}^{i-1}P_j-\sum\limits_{j=i+1}^{M}P_j}.\label{eqGP}
\end{align}
Finally, $(d)$ comes from 	 averaging over the fading $g_j$, attaining  \eqref{eqP}. By isolating $\gamma_i$ in \eqref{eqGP} we reach \eqref{gami}. Notice that \eqref{eqP} may have up to $|\Phi\backslash\{b_0\}|$ solutions; however $\gamma_i>0$ then $\varphi_i^*$ has to be also real positive according to \eqref{gami}. Now, the left term of \eqref{gami} is monotonically  increasing on $\varphi_i$ and both left and right terms share the same domain $(0,\infty)$ $\forall\varphi_i>0,0<\epsilon_i<1$, hence  \eqref{gami} has only one real positive solution $\varphi_i^*$.
  \hfill\qedsymbol
\section{Proof of Theorem~\ref{the2}}\label{App_B}
Let's consider that the number of interfering BSs is finite and equal to $n$ and we sort in ascending order the interfering BSs according to their distance to $\mathrm{UE}_i$. Hence, $r_{j,i}, j\ge 1,$ denotes the distance from $\mathrm{UE}_i$ to its $j-$nearest interfering BS. Notice that by letting $n\rightarrow\infty$ we are also able to model an infinite network deployment. Now, by using the relation between the geometric and the arithmetic mean we have that
\begin{align}
\Big(\prod_{j=1}^n(1+\varphi_i r_i^{\alpha}r_{j,i}^{-\alpha})\Big)^{\frac{1}{n}}&\le \frac{1}{n}\sum_{j=1}^{n}(1+\varphi_i r_i^{\alpha}r_{j,i}^{-\alpha})\nonumber\\	
\prod_{j=1}^{n}(1+\varphi_i r_i^{\alpha}r_{j,i}^{-\alpha})&\le\bigg[\frac{1}{n} (n+\varphi_i r_i^{\alpha}\sum_{j=1}^{n}r_{j,i}^{-\alpha})\bigg]^{n}\nonumber\\
\prod_{j=1}^{n}(1+\varphi_i r_i^{\alpha}r_{j,i}^{-\alpha})&\le\bigg[ 1+\frac{\varphi_i r_i^{\alpha}}{n}\sum_{j=1}^{n}r_{j,i}^{-\alpha}\bigg]^{n}.\label{mean}	
\end{align}
\begin{figure}[t!]
	\includegraphics[width=0.47\textwidth,right]{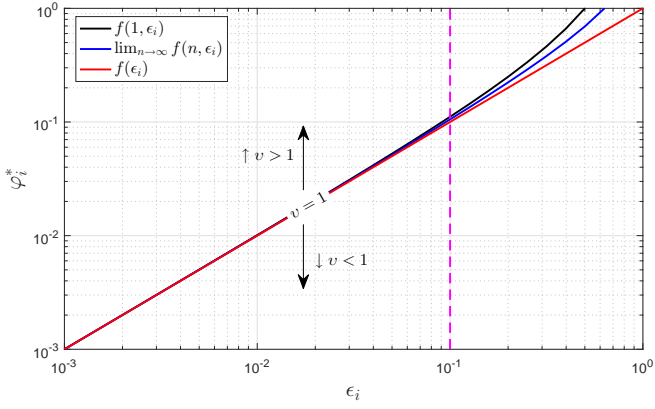}
	\caption{$\varphi_i^*$ as a function of $\epsilon_i$ and according to \eqref{ap}. We set $\upsilon=r_i^{-\alpha}/\sum\limits_{j=1}^{n}r_{j,i}^{-\alpha}$.}
	\label{Fig9}
\end{figure}
By using the right term of \eqref{mean} instead of the left one when solving \eqref{eqP} for $\varphi_i$ yields
\begin{align}
\varphi_i^*\approx\frac{n r_i ^{-\alpha}}{\sum\limits_{j=1}^{n}r_{j,i}^{-\alpha}}\Big[\Big(\frac{1}{1-\epsilon_i}\Big)^{\frac{1}{n}}-1\Big].\label{appB}
\end{align}
By looking at \eqref{eqP} notice that for small $\epsilon_i$, e.g., $\epsilon_i\le 10^{-1}$, $\frac{1}{1-\epsilon_i}$ becomes closer to $1$ and since each of the terms $(1+\hat{\varphi}_i r_i^{\alpha}r_{j,i}^{-\alpha})$ are strictly greater that $1$, we can expect that each of them approximates heavily to the unity, hence, all of these terms are very similar between each other. Since geometric and arithmetic means match when all the operating terms are equal, we can be sure that for small $\epsilon_i$, e.g., $\epsilon_i\le 10^{-1}$, the approximation given in \eqref{appB} is accurate.

Now we can write $\varphi_i^*$ as
\begin{align}
\varphi_i^*\approx\frac{r_i ^{-\alpha}}{\sum\limits_{j=1}^{n}r_{j,i}^{-\alpha}}f(n,\epsilon_i),\label{ap}
\end{align}
where $f(n,\epsilon_i)\!=\!n\big((1\!-\!\epsilon_i)^{-\frac{1}{n}}\!-\!1\big)$ is a decreasing function of $n$ since $\frac{\partial f(n,\epsilon_i)}{\partial n}\!=\!\frac{1}{(1\!-\!\epsilon_i)^n}\!+\!\frac{\ln(1\!-\!\epsilon_i)}{n(1\!-\!\epsilon_i)^n}\!-\!1\!<\!0$, $\forall \epsilon_i: 0\!<\!\epsilon_i\!<\!1$. Therefore, its maximum and minimum values are when $n=1$ and $n\rightarrow\infty$, respectively, for which $f(1,\epsilon_i)=\epsilon_i/(1-\epsilon_i)$ and $\lim_{n\rightarrow\infty}f(n,\epsilon_i)=-\ln(1-\epsilon_i)$. Interestingly, when $\epsilon_i\le 10^{-1}$ both
results are very close to each other, which leads to the conclusion that $f(n,\epsilon_i)$ depends very weakly on $n$ in that region. In fact, following result holds
\begin{align}
f(n,\epsilon_i)\approx f(\epsilon_i)=\epsilon_i\label{f}
\end{align}
hence $\varphi_i^*$ can be written as in \eqref{app}. The accuracy of \eqref{f} when $\epsilon_i\le 10^{-1}$ can be easily appreciated in Fig.~\ref{Fig9}. Notice that the value of $\upsilon=r_i^{-\alpha}/\sum\limits_{j=1}^{n}r_{j,i}^{-\alpha}$ only causes a vertical shift of the curve $\varphi_i^*$ vs $\epsilon_i$.
 \hfill\qedsymbol
\section{Proof of Theorem~\ref{the3}}\label{App_C}
Using \eqref{app} we proceed as follows
\begin{align}
F_{\varphi_i^*}(\omega)&=\mathbb{P}\Bigg(\frac{r_i ^{-\alpha}}{\sum\limits_{j\in\Phi\backslash\{b_0\}}r_{j,i}^{-\alpha}}\epsilon_i<\omega\Bigg)\nonumber\\
&=\mathbb{P}\Bigg(\sum\limits_{j\in\Phi\backslash\{b_0\}}\Big(\frac{r_i}{r_{j,i}}\Big)^{\alpha}>\frac{\epsilon_i}{\omega}\Bigg)\nonumber\\
&=\mathbb{P}\Bigg(\sum\limits_{x\in\Phi\backslash\{x: ||x||<r_i\}}\Big(\frac{r_i}{||x||}\Big)^{\alpha}>\frac{\epsilon_i}{\omega}\Bigg)\nonumber\\
&=1-F_{\Psi}\Bigg(\frac{\epsilon_i}{\omega}\Bigg),\label{f1}
\end{align}
where $\Psi=\sum\limits_{x\in\Phi\backslash\{x: ||x||<r_i\}}\big(\frac{r_i}{||x||}\big)^{\alpha}$. Similar to \cite[Def. 1]{Ganti2.2016}	we define the squared relative distance process (SRDP) of the point process $\Phi$ as $\Theta=\big\{x\in\Phi\backslash\{x: ||x||<r_i\}: \frac{||x||^2}{r_i^2}\big\}$, then
\begin{align}
\Psi=\sum\limits_{x\in\Theta}x^{-1/\delta},\label{Psi}
\end{align}
with $\delta=2/\alpha$. 
Now, the probability generating functional (PGFL) of the SRDP is given by
\begin{align}
G_{\Theta}[f]&=\mathbb{E}\bigg[\prod_{x\in\Theta}f(x)\bigg]=\mathbb{E}\bigg[\prod_{x\in\Phi}f\Big(\frac{||x||^2}{r_i^2}\Big)\bigg]\nonumber\\
&\stackrel{(a)}{=}\int\limits_0^{\infty}\exp\Big(\!\!-\!2\lambda\pi\int\limits_{r_i}^{\infty}\eta(1-f(\eta^2/r_i^2))\mathrm{d}\eta\Big)f_{r_i}(r_i)dr_i\nonumber\\
&\stackrel{(b)}{=}2\pi\lambda\int\limits_0^{\infty}r_i\exp\bigg(\!\!-\!\lambda\pi r_i^2\Big(1+\int\limits_{1}^{\infty}(1-f(y))\mathrm{d}y\Big)\bigg)dr_i\nonumber\\
&\stackrel{(c)}{=}\frac{1}{1\!+\!\int_{1}^{\infty}\!\big(1\!-\!f(y)\big)dy},\label{PGFL}
\end{align}
where $(a)$ follows from the PGFL of the PPP, $(b)$ is obtained by the substitution $\eta^2/r_i^2\rightarrow y$ and using the distribution of the nearest neighbor distance, while $(c)$ follows from solving the outer integral. 
%
%
%
	Then,
	\begin{align}
	\mathcal{L}_{\Psi}(s)&=\frac{1}{1+\int_{1}^{\infty}\Big(1-\exp\big({-sx^{-1/\delta}}\big)\Big)dy}\nonumber\\
	&=\frac{1}{\ _1F_1(-\delta,1-\delta,-s)},
	\end{align}
	%
\begin{align}
F_{\Psi}(x)&=\mathcal{L}^{-1}\bigg\{\frac{1}{s}\mathcal{L}_{\Psi}(s)\bigg\}(x)\nonumber\\
&=\mathcal{L}^{-1}\bigg\{\frac{1}{s\ _1F_1(-\delta,1-\delta,-s)}\bigg\}(x),\label{f2}
\end{align}
while plugging \eqref{f2} into \eqref{f1} and then using the variable transformation given in \eqref{eqGP} we attain \eqref{eqth3}. Notice that although \eqref{eqth3} is general, it is not in closed-form. However, by noticing that $\varphi_i^*$ in \eqref{app} has the form of the SIR of a typical link in a Poisson network without fading, but scaled by the reliability constraint $\epsilon_i$, we can take advantage of \cite[Eq. 6]{Zhang.2014} to write
\begin{align}
F_{\varphi_i^*}(\omega)&=1-\mathrm{sinc}(\delta)\Big(\frac{\omega}{\epsilon_i}\Big)^{-\delta},\qquad \mathrm{for}\ \omega\ge\epsilon_i,\label{bound}
\end{align}
which also works as a lower-bound when $\omega<\epsilon_i$, and for the special case of $\alpha=4$ yields $F_{\varphi_i^*}(\omega)=1-1/\big(\Gamma(1.5)\sqrt{\pi\omega/\epsilon_i}\big)$. Finally, using \eqref{bound} along with the variable transformation given in \eqref{eqGP} we attain \eqref{eqth4}. \hfill\qedsymbol
\section{Proof of Theorem~\ref{the4}}\label{App_D}
Using \eqref{gami} and the fact that $\gamma_1=\gamma_2$ we have that
\begin{align}
\frac{\varphi_1^*\beta}{1+(1-\beta)\varphi_1^*}-\frac{\varphi_2^*(1-\beta)}{1+\beta\mu\varphi_2^*}&=0\nonumber\\
\varphi_1^*\varphi_2^*(1\!-\!\mu)\beta^2\!-\!(\varphi_1^*\!+\!\varphi_2^*\!+\!2\varphi_1^*\varphi_2^*)\beta\!+\!\varphi_2^*(1\!+\!\varphi_1^*)\!&=\!0,
\end{align}
and \eqref{bet1} comes easy by solving this quadratic equation for $\beta$ while discarding the solution that leads to $\beta>1$. The reader can check that $\beta$ in \eqref{bet1} always lies in $[0,\ 1]$. However, the other necessary condition for feasibility is that $\beta\ge 1/2$, thus
\begin{align}
\frac{\varphi_1^*+\varphi_2^*+2\varphi_1^*\varphi_2^*}{2(1-\mu)\varphi_1^*\varphi_2^*}+\qquad\qquad\qquad\qquad\qquad\qquad\qquad &\nonumber\\
-\frac{\sqrt{(\varphi_1^*+\varphi_2^*)^2+4\varphi_1^*\varphi_2^*(\varphi_1^*+\mu\varphi_2^*+\mu\varphi_1^*\varphi_2^*)}}{2(1-\mu)\varphi_1^*\varphi_2^*}&\ge\frac{1}{2}\nonumber\\
\varphi_1^*\varphi_2^*\!-\!2\mu\varphi_1^*\varphi_2^*\!+\!\mu^2\varphi_1^*\varphi_2^*\!-\!2\varphi_1^*\!+\!2\mu\varphi_1^*\!+\!2\varphi_2^*\!-\!2\mu\varphi_2^*\!&\ge\! 0\nonumber\\
\qquad\qquad \frac{2\varphi_1^*(1-\mu)}{2-2\mu+\varphi_1^*-2\mu\varphi_1^*+\mu^2\varphi_1^*}&\le \varphi_2^*\nonumber\\
\qquad\qquad \frac{2\varphi_1^*(1-\mu)}{2(1-\mu)+\varphi_1^*(1-\mu)^2}&\le\varphi_2^*\nonumber\\
\qquad\qquad\frac{2\varphi_1^*}{2+\varphi_1^*(1-\mu)}&\le  \varphi_2^*.\label{str}
\end{align}
Notice that when $\mu$ increases, the required $\varphi_2^*$ to ensure similar rate performance for both UEs, increases. This is, the worse the SIC performs, the stringent the resource reservation for $\mathrm{UE}_2$. This is also aligned with the fact that $\beta$ in \eqref{bet1} is a decreasing function of $\mu$.

Continuing with the proof, $\varphi_1^*$ and $\varphi_2^*$ are independent since each of them only depends on the distances from the UEs to the transmitting and interfering BSs, path-loss exponent and reliability constraints of the UEs (See \eqref{eqP}). Let the UEs with decoding order 1 and 2, be denoted as $A$ and $B$, respectively, thus, $\varphi_1^*=\varphi_A^*$ and $\varphi_2^*=\varphi_B^*$. If this assignment already satisfies \eqref{str}, e.g., $\varphi_B^*\ge\frac{2\varphi_A^*}{2+\varphi_A^*(1-\mu)}$, then, the power allocation profile given in \eqref{bet1} is feasible and both UEs can attain the same rate performance. In the opposite case we have that
	\begin{align}
	\varphi_B^*&<\frac{2\varphi_A^*}{2+\varphi_A^*(1-\mu)}\nonumber\\
	\varphi_A^*&>\frac{2\varphi_B^*}{2-\varphi_B^*(1-\mu)}\ge\frac{2\varphi_B^*}{2+\varphi_B^*(1-\mu)}.
	\end{align}
	Therefore, if we scheduled the UEs such that $\varphi_1^*=\varphi_B^*$ and $\varphi_2^*=\varphi_A^*$ we can also ensure the same rate performance for both UEs.  \hfill\qedsymbol
\section{Proof of Corollary~\ref{cor2}}\label{App_E}
By setting $S=\varphi_1^*+\varphi_2^*$ and $P=\varphi_1^*\varphi_2^*$, which are fixed-value terms and independent of the decoding order, we can write \eqref{gam} as
\begin{align}
\gamma_1=\frac{\sqrt{S^2+4Pk}-S}{2k},
\end{align}
where $k=\varphi_1^*+\mu\varphi_2^*+\mu P$. Now we have that
\begin{align}
\frac{d\gamma_1}{dk}&=\frac{S\sqrt{S^2+4Pk}-(S^2+2PK)}{2k^2\sqrt{S^2+4Pk}}<0,
\end{align}
which is negative since
\begin{align}
4P^2k^2&> 0\nonumber\\
4P^2k^2+S^4+4PS^2k&> S^4+4PS^2k\nonumber\\
(S^2+2Pk)^2&> S^2(S^2+4Pk)\nonumber\\
S^2+2Pk&>S\sqrt{S^2+4Pk},
\end{align}
therefore, $\gamma_1$ is a decreasing function of $k$ and gets the greater value for the smaller $k$, which occurs when we choose $\varphi_2^*\ge \varphi_1^*$. 
Also, notice that with such scheduling the operation is guaranteed since \eqref{str} is satisfied $\big(\varphi_2^*\ge \varphi_1^*\ge \frac{2\varphi_1^*}{2+\varphi_1^*(1-\mu)}\big)$.  \hfill\qedsymbol
\section{Proof of Corollary~\ref{pro1}}\label{App_F}
We proceed as follows
\begin{align}
\tilde{\gamma}=(1+\gamma_1)(1+\gamma_2)&\stackrel{(a)}{\lesssim}\bigg[\frac{(1+\gamma_1)+(1+\gamma_2)}{2}\bigg]^2\nonumber\\
&= \Big[1+\frac{\gamma_1+\gamma_2}{2}\Big]^2, \label{gb}
\end{align}
where $(a)$ comes from using the inequality  between arithmetic and geometric means. By using \eqref{gb} as an approximation and setting $\bar{\gamma}=\gamma_1+\gamma_2$ we attain \eqref{gamB}.

Now let's present some insights on the accuracy of \eqref{gamB}. The relative error function
\begin{align}
\xi(\%)&=100\times\frac{\Big[1+\frac{1}{2}\bar{\gamma}\Big]^2-\tilde{\gamma}}{\tilde{\gamma}}\nonumber\\
&=25\times\Big[\frac{1+\gamma_1}{1+\gamma_2}+\frac{1+\gamma_2}{1+\gamma_1}\Big]-50,
\end{align} 
measures the error when adopting the approximation given in \eqref{gamB}, and it is plotted in Fig.~\ref{Fig10}. As expected, the approximation given in \eqref{gamB} is more accurate when $\gamma_1\approx\gamma_2$ and/or $\gamma_1,\gamma_2\ll 1$. 
The latter condition is typical of scenarios where devices are with low transmission rates because of the short length of their packets and/or stringent reliability constraints, e.g., MTC setups. In those scenarios $\max\limits_{\beta}\tilde{\gamma}\approx\max\limits_{\beta}\bar{\gamma}$.  \hfill\qedsymbol
\begin{figure}[t!]
	\includegraphics[width=0.47\textwidth,right]{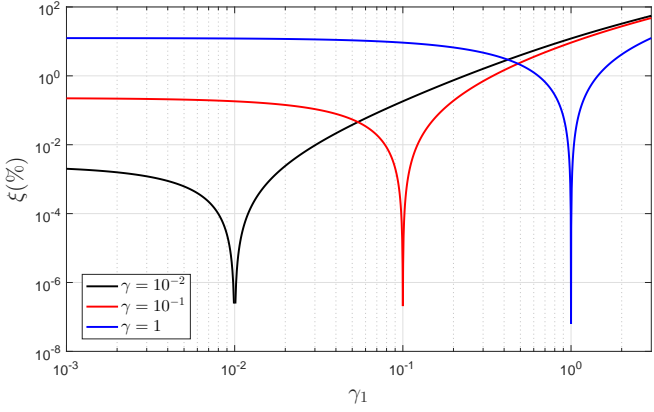}
	\caption{Relative error function, $\xi (\%)$, as a function of $\gamma_1$ with $\gamma_2\in\{10^{-2},10^{-1},1\}$.}		
	\label{Fig10}
\end{figure}
\section{Proof of Theorem~\ref{the5}}\label{App_G}
	We have that
\begin{align}
\bar{\gamma}&=\gamma_1+\gamma_2,\nonumber\\
\frac{d\bar{\gamma}}{d\beta}&=\frac{d\gamma_1}{d\beta}+\frac{d\gamma_2}{d\beta}\nonumber\\
&=\frac{\varphi_1^*(1+\varphi_1^*)}{(1+\varphi_1^*(1-\beta))^2}-\frac{\varphi_2^*(1+\mu\varphi_2^*)}{(1+\mu\beta\varphi_2^*)^2},\\
\frac{d^2\bar{\gamma}}{d\beta^2}&=\frac{d^2\gamma_1}{d\beta^2}+\frac{d^2\gamma_2}{d\beta^2}\nonumber\\
&=\frac{2\varphi_1^{*2}(1+\varphi_1^*)}{(1+(1-\beta)\varphi_1^*)^3}+\frac{2\mu\varphi_2^{*2}(1+\mu\varphi_2^*)}{(1+\mu\beta\varphi_2^*)^3}.
\end{align}	
Since $\frac{d^2\bar{\gamma}}{d\beta^2}>0\ \forall\beta$, then $\bar{\gamma}$ is convex on $\beta$ and a global minimum exists. That means that the maximum lies on the interval extremes, $\beta^*=1/2$ or $\beta^*=1$. Thus,
\begin{align}
\max\limits_{\beta}\bar{\gamma}&=\max\Big(\bar{\gamma}|_{\beta=\tfrac{1}{2}},\bar{\gamma}|_{\beta=1}\Big)\nonumber\\
&=\max\bigg(\frac{\frac{1}{2}\varphi_1^*}{1+\frac{1}{2}\varphi_1^*}+\frac{\frac{1}{2}\varphi_2^*}{1+\frac{1}{2}\mu\varphi_2^*},\varphi_1^*\bigg)\nonumber\\
&=\max\bigg(\frac{\varphi_1^*}{2+\varphi_1^*}+\frac{\varphi_2^*}{2+\mu\varphi_2^*},\varphi_1^*\bigg).
\end{align}
Then,
\begin{align}
\frac{\varphi_1^*}{2+\varphi_1^*}+\frac{\varphi_2^*}{2+\mu\varphi_2^*}&\stackequal{\beta^*=1/2}{\beta^*=1}\varphi_1^*\nonumber\\
\frac{2(\varphi_2^*-\varphi_1^*)+\varphi_1^*(\varphi_2^*-2)}{\varphi_1^*\varphi_2^*(1+\varphi_1^*)}&\stackequal{\beta^*=1/2}{\beta^*=1}\mu,
\end{align}
which implies \eqref{eqBo}.  \hfill\qedsymbol
\section{Proof of Theorem~\ref{the6}}\label{App_H}
For the maximum sum-rate setup, the NOMA scheme overcomes the OMA configuration when
\begin{align}
\log_2\tilde{\gamma}&>\frac{1}{2}\log_2(1+\varphi_1^*)+\frac{1}{2}\log_2(1+\varphi_2^*)\nonumber\\
&=\log_2\sqrt{(1+\varphi_1^*)(1+\varphi_2^*)}\nonumber\\
\tilde{\gamma}&>\sqrt{(1+\varphi_1^*)(1+\varphi_2^*)}.\label{gamCor}
\end{align}
We know that $\varphi_2^*\ge\varphi_1^*$, and first we are going to work on that to proof the statement of the theorem for $\mu=0$. We proceed as follows
\begin{align}
\varphi_2^*&\ge \varphi_1^*\qquad|\times\varphi_1^*,\ \ \  (\varphi_1^*\!>\!0)\!\nonumber\\
\varphi_1^*\varphi_2^*&\ge \varphi_1^{*2}\nonumber\\
\!(\!4\!+\!2\varphi_2^*\!+\!4\varphi_1^*\!+\!\varphi_1^*\varphi_2^*)\!+\!  \varphi_1^*\varphi_2^*\!&\ge\! (\!4\!+\!2\varphi_2^*\!\!+\!4\varphi_1^*\!+\!\varphi_1^*\varphi_2^*)\!+\!\varphi_1^{*2}\!\nonumber\\
2(1+\varphi_1^*)(2+\varphi_2^*)&\ge(2+\varphi_1^*+\varphi_2^*)(2+\varphi_1^*)\nonumber\\
\frac{(1+\varphi_1^*)(2+\varphi_2^*)}{2+\varphi_1^*}&\ge\frac{2+\varphi_1^*+\varphi_2^*}{2}\nonumber\\
\frac{(1+\varphi_1^*)(2+\varphi_2^*)}{2+\varphi_1^*}&\ge\frac{(1\!+\!\varphi_1^*)+(1\!+\!\varphi_2^*)}{2}.
\end{align}
Notice that the left term matches the expression for $\tilde{\gamma}|_{\mu=0}$ given in \eqref{bmu0} for $\beta=1/2$ and using the inequality between the arithmetic and geometric means we have that $\frac{(1\!+\!\varphi_1^*)+(1\!+\!\varphi_2^*)}{2}\ge \sqrt{(1\!+\!\varphi_1^*)(1\!+\!\varphi_2^*)}$, thus,
\begin{align}
\tilde{\gamma}|_{\mu=0}\ge \sqrt{(1\!+\!\varphi_1^*)(1\!+\!\varphi_2^*)},
\end{align}
where the right term matches the performance of the OMA setup. Therefore, for $\mu=0$ the NOMA scheme is the best choice.

For $\mu>0$ we perform some simplifications to gain in mathematical tractability. 
We can establish that the NOMA scheme will overcome \textit{almost surely} the OMA configuration if
\begin{align}
\Big[1+\frac{1}{2}\bar{\gamma}\Big]^2&>1+\frac{\varphi_1^*+\varphi_2^*}{2}\nonumber\\
\Big[1+\frac{\frac{1}{2}\varphi_1^*}{2+\varphi_1}+\frac{\frac{1}{2}\varphi_2^*}{2+\mu\varphi_2^*}\Big]^2&>1+\frac{\varphi_1^*+\varphi_2^*}{2},\label{eqMu}
\end{align}
where the first line comes from using \eqref{gamB} and using the inequality relating the arithmetic and geometric means for the right term ($\sqrt{(1+\varphi_1^*)(1+\varphi_2^*)}$, which is the achievable rate with OMA, is greater than or equal to $1+\frac{\varphi_1^*+\varphi_2^*}{2}$), and the second line comes from using \eqref{gamB2}. Now, solving \eqref{eqMu} as a function of $\mu$ we attain \eqref{cor4}.  \hfill\qedsymbol

\bibliographystyle{IEEEtran}
\bibliography{IEEEabrv,references}
\end{document}